\documentclass[acmsmall,screen]{acmart}
\pdfoutput=1

\makeatletter                   
\def\mdseries@tt{m}             
\makeatother                    

\usepackage{booktabs}   
\usepackage{subcaption} 
\usepackage{latexsym}
\usepackage{setspace}
\usepackage{cancel}
\usepackage{listings}
\usepackage{graphicx}
\usepackage{appendix}
\usepackage{amssymb}
\usepackage{stmaryrd}
\usepackage{amsmath}
\usepackage{leftidx}
\usepackage{mathtools}
\usepackage{paralist}
\usepackage{color}
\usepackage{mathrsfs}
\usepackage{tikz}
\usepackage[draft]{minted}
\usetikzlibrary{shapes}
\usepackage[linesnumbered,ruled]{algorithm2e}

\newcommand{\Nat}{\ensuremath{\mathbb{N}}}

\newcommand{\eqdef}{\stackrel{\mbox{\begin{tiny}def\end{tiny}}}{=}} 

\newcommand{\cut}[1]{}

\mathchardef\mhyphen="2D 

\newcommand{\hide}[1]{}

\newcommand\cA{\mathcal{A}}
\newcommand\cB{\mathcal{B}}

\newcommand\cE{\mathcal{E}}
\newcommand\cG{\mathcal{G}}
\newcommand\Ll{\mathcal{L}}

\newcommand\cP{\mathcal{P}}

\newcommand\cT{\mathcal{T}}

\newcommand\replaceall{\mathsf{replaceAll}}
\newcommand\indexof{\mathsf{IndexOf}}

\newcommand\strline{\mathsf{SL}}

\newcommand\search{\mathsf{search}}

\newcommand\verify{\mathsf{vfy}}

\newcommand\searchleft{\mathsf{left}}

\newcommand\searchlong{\mathsf{long}}

\newcommand\wprof{\mathsf{WP}}

\newcommand\vars{\mathsf{Vars}}

\newcommand\rpleft{\mathsf{l}}
\newcommand\rpright{\mathsf{r}}

\newcommand\red{\mathsf{red}}

\newcommand{\dmdidx}{{\sf Idx_{dmd}}}
\newcommand{\lftlen}{{\sf Len_{lft}}}

\newcommand{\ASSERT}[1]{\textbf{assert}(#1)}

\newcommand{\OMIT}[1]{}

\newcommand\shortlong[2]{#2}

\newif\ifdraft\draftfalse
\ifdraft
\newcommand{\anthony}[1]{\color{red} {YA: #1 :AY} \color{black}}
\newcommand{\zhilin}[1]{\color{brown} {ZL: #1 :LZ} \color{black}}
\newcommand{\tl}[1]{\color{blue} {TL: #1 :LT} \color{black}}
\newcommand{\mat}[1]{\color{cyan} {MH: #1 :HM} \color{black}}
\else
\newcommand{\anthony}[1]{}
\newcommand{\zhilin}[1]{}
\newcommand{\tl}[1]{}
\newcommand{\mat}[1]{}
\fi

\newcommand{\concat} {\circ}

\newcommand{\str} {{\sf Str}}
\newcommand{\intnum} {{\sf Int}}
\newcommand{\regexp} {{\sf RegExp}}

\setcopyright{acmlicensed}
\acmPrice{}
\acmDOI{10.1145/3158091}
\acmYear{2018}
\copyrightyear{2018}
\acmJournal{PACMPL}
\acmVolume{2}
\acmNumber{POPL}
\acmArticle{3}
\acmMonth{1}

\begin{document}
\newtheorem{remark}[theorem]{Remark}

\title[What Is Decidable about String Constraints with the ReplaceAll
Function]{What Is Decidable about String Constraints with the ReplaceAll
Function \shortlong{}{(Technical Report)}}         


\author{Taolue Chen}
\orcid{0000-0002-5993-1665}
\affiliation{
  \department{Department of Computer Science and Information Systems}              
  \institution{Birkbeck, University of London}            
  \streetaddress{Malet Street}
  \city{London}
  \postcode{WC1E 7HX}
  \country{United Kingdom}
}
\email{taolue@dcs.bbk.ac.uk}          

\author{Yan Chen}
\affiliation{
  \institution{State Key Laboratory of Computer Science, Institute of Software, Chinese Academy of Sciences} 
  \country{China}
}
\affiliation{
  \institution{University of Chinese Academy of Sciences} 
  \country{China}
}

\author{Matthew Hague}
\orcid{0000-0003-4913-3800}
\affiliation{
  \department{Department of Computer Science} 
  \institution{Royal Holloway, University of London} 
  \streetaddress{Egham Hill}
  \city{Egham}
  \state{Surrey}
  \postcode{TW20 0EX}
  \country{United Kingdom}
}
\email{matthew.hague@rhul.ac.uk}          

\author{Anthony W. Lin}
\orcid{0000-0003-4715-5096}
\affiliation{
  \department{Department of Computer Science}              
  \institution{University of Oxford}            
  \streetaddress{Wolfson Buildin, Parks Road}
  \city{Oxford}
  \postcode{OX1 3QD}
  \country{United Kingdom}
}
\email{anthony.lin@cs.ox.ac.uk}          

\author{Zhilin Wu}
\affiliation{
  \institution{State Key Laboratory of Computer Science, Institute of Software, Chinese Academy of Sciences} 
  \country{China}
}




\begin{abstract}

The theory of strings with concatenation has been widely argued as the basis of
constraint solving for verifying string-manipulating programs. However, this
theory is far from adequate for expressing many string constraints that are
also needed in practice; for example, the use of regular constraints (pattern matching
against a regular expression), and the string-replace function (replacing
either the first occurrence or all occurrences of a ``pattern'' string
constant/variable/regular expression by a ``replacement'' string
constant/variable), among many others. Both regular constraints and the
string-replace function are crucial for such applications as analysis of
JavaScript (or more generally HTML5 applications) against cross-site scripting
(XSS) vulnerabilities, which motivates us to consider a richer class of string
constraints. The importance of the string-replace function (especially the
replace-all facility) is increasingly recognised, which can be witnessed by the
incorporation of the function in the input languages of several string
constraint solvers. 

Recently, it was shown that any theory of strings containing the string-replace
function (even the most restricted version where pattern/replacement strings
are both constant strings) becomes undecidable if we do not impose some kind of
straight-line (aka acyclicity) restriction on the formulas. Despite this,
the straight-line restriction is still practically sensible since this condition is  
typically met by string constraints that are generated by symbolic   execution.
In this paper, we provide the first systematic study of straight-line string 
constraints with the string-replace function and the regular constraints as the 
basic operations. We show that a large class of such constraints (i.e. when
only a constant string or a regular expression is permitted in the
pattern) is decidable. We note that the string-replace function, even under
this restriction, is sufficiently powerful for expressing the concatenation
    operator and much more (e.g. extensions of regular expressions with string variables).
This gives us the most expressive decidable logic containing concatenation,
replace, and regular constraints under the same umbrella.
Our decision procedure for the straight-line fragment follows an
    automata-theoretic approach, and is modular in the sense that the string-replace terms are removed one by one to generate more and more regular constraints, which can then be discharged by the state-of-the-art string constraint solvers. 
We also show that this fragment is, in a way, a maximal decidable subclass of
the straight-line fragment with string-replace and regular constraints.
To this end, we show undecidability results for the following two
extensions: (1) variables are permitted in the pattern parameter of
    the replace function, (2) length constraints are permitted.
\end{abstract}

\maketitle


\section{Introduction}
\label{sec:intro}

The problem of 
automatically solving string constraints (aka satisfiability of logical theories over
strings) has recently witnessed renewed interests 
\cite{Berkeley-JavaScript,TCJ16,LB16,YABI14,S3,Abdulla14,Abdulla17,DV13,symbolic-transducer,BEK,HAMPI,cvc4,Z3-str,fang-yu-circuits,BTV09} 
because of important applications in the analysis of string-manipulating  programs. For example,
program analysis techniques like symbolic execution
\cite{king76,DART,EXE,jalangi} 
would
systematically explore executions in a program and collect symbolic path 
constraints, which could then be solved using a constraint solver and
used to determine which location in the program to continue exploring.
To successfully apply a constraint solver in this instance, it is
crucial that the constraint language precisely models the data types in the
program, along with the data-type operations used. In the context of
string-manipulating programs, this could include 
concatenation, regular constraints (i.e. pattern matching against a regular
expression), string-length functions, and the string-replace functions, among 
many others.

Perhaps the most well-known theory of strings for such applications as the
analysis of string-manipulating programs is the theory of strings with concatenation  (aka \emph{word equations}), whose decidability was shown by Makanin \cite{Makanin} in 1977 after it was open for many years. More importantly, 
this theory remains decidable even when regular constraints are incorporated into the 
language \cite{Schulz}. However, whether adding the string-length function
 preserves the decidability remains a long-standing open problem
\cite{Vijay-length,buchi}. 

Another important string operation---especially in popular scripting
languages like Python, JavaScript, and PHP---is the \emph{string-replace function}, 
which may be used to replace either the \emph{first} occurrence or
\emph{all} occurrences of a string (a string constant/variable, or a regular expression) by 
another string (a string constant/variable). The replace function (especially 
the replace-all functionality) is omnipresent in HTML5 applications
\cite{LB16,TCJ16,YABI14}. 
For example, a standard industry defense against cross-site scripting 
(XSS) vulnerabilities includes sanitising untrusted strings before adding them
into the DOM (Document Object Model) or the HTML document. 
This is typically done by 
various metacharacter-escaping mechanisms (see, for instance, 
\cite{Kern14,BEK,OWASP-XSS}). An example of such a mechanism is backslash-escape, which replaces \emph{every
occurrence} of quotes and double-quotes (i.e. \verb+'+ and \verb+"+) in the
string by \verb+\'+ and \verb+\"+. 
In addition to sanitisers, common JavaScript functionalities like \texttt{document.write()} 
and \texttt{innerHTML} apply an \emph{implicit browser transduction} --- which
decodes HTML codes (e.g. \verb+&#39;+ is replaced by \verb+'+) in the input 
string --- before inserting the input string into the DOM.
Both of these examples can be expressed by (perhaps multiple) 
applications of the string-replace function.
Moreover, although these examples replace constants by constants, the popularity of template systems such as Mustache~\cite{Mustache} and Closure Templates~\cite{Closure} demonstrate the need for replacements involving variables.
Using Mustache, a web-developer, for example, may define an HTML fragment with placeholders that is instantiated with user data during the construction of the delivered page.





\begin{example}
We give a simple example demonstrating a (naive) XSS vulnerability to illustrate the use of string-replace functions.
Consider the HTML fragment below.
\begin{minted}{html}
   <h1> User <span onMouseOver="popupText('{{bio}}')">{{userName}}</span> </h1>
\end{minted}
This HTML fragment is a template as might be used with systems such as Mustache to display a user on a webpage.
For each user that is to be displayed -- with their username and biography stored in variables \emph{user} and \emph{bio} respectively -- the string \verb+{{userName}}+ will be replaced by \emph{user} and the string \verb+{{bio}}+ will be replaced by \emph{bio}.
For example, a user \verb+Amelia+ with biography \verb+Amelia was born in 1979...+ would result in the HTML below.
\begin{minted}{html}
   <h1> User 
        <span onMouseOver="popupText('Amelia was born in 1979...')">
            Amelia </span> </h1>
\end{minted}
This HTML would display \verb+User Amelia+, and, when the mouse is placed over \verb+Amelia+, her biography would appear, thanks to the \verb+onMouseOver+ attribute in the \verb+span+ element.

Unfortunately, this template could be insecure if the user biography is not adequately sanitised: 
A user could enter a malicious biography, such as \verb+'); alert('Boo!'); alert('+ which would cause the following instantiation of the \verb+span+ element\footnote{
	Readers familiar with Mustache and Closure Templates may expect single quotes to be automatically escaped.
	However, we have tested our example with the latest versions of mustache.js~\cite{MustacheJS} and Closure Templates~\cite{Closure} (as of July 2017) and observed that the exploit is not disarmed by their automatic escaping features.
}.
\begin{minted}{html}
    <span onMouseOver="popupText(''); alert('Boo!'); alert('')">
\end{minted}
Now, when the mouse is placed over the user name, the malicious JavaScript \verb+alert('Boo!')+ is executed.

The presence of such malicious injections of code can be detected using string constraint solving and XSS \emph{attack patterns} given as regular expressions~\cite{BCFJKKV08,Berkeley-JavaScript,YABI14}.
For our example, given an attack pattern $P$ and template $temp$, we would generate the constraint
\[
    x_1 = \replaceall(temp, \verb+{{userName}}+, \mathit{user})
    \land
    x_2 = \replaceall(x_1, \verb+{{bio}}+, \mathit{bio})
    \land
    x_2 \in P
\]
which would detect if the HTML generated by instantiating the template is
    susceptible to the attack identified by $P$. \qed
\end{example}

In general, the string-replace function has three parameters, and in the current mainstream language such as Python and JavaScript, \emph{all of the three parameters can be inserted as string variables}. As result, when we perform program analysis for, for instance, detecting security vulnerabilities as described above, one often obtains string constraints of the form $z= \replaceall(x, p, y)$, where $x,y$ are string constants/variables, and $p$ is either a string constant/variable or a regular expression.
Such a constraint means that $z$ is obtained by replacing all occurrences of $p$
in $x$ with $y$. For convenience, we call $x, p, y$ as the \emph{subject}, the
\emph{pattern}, and the \emph{replacement} parameters respectively. 




The $\replaceall$ function is a powerful string operation that goes beyond the 
expressiveness of concatenation. (On the contrary, as we will see later, concatenation can be expressed by the $\replaceall$ function easily.) 
It was shown in a recent POPL paper \cite{LB16} that any theory of strings containing the 
string-replace function (even the most restricted version where 
pattern/replacement strings are both constant strings) becomes undecidable if 
we do not impose some kind of
\emph{straight-line restriction}\footnote{Similar notions that appear in the 
literature of string constraints (without replace) include acyclicity 
\cite{Abdulla14} and solved form \cite{Vijay-length}} on 
the formulas. Nonetheless, as already noted in \cite{LB16},
the straight-line restriction is reasonable since it
is typically satisfied by constraints that are generated by symbolic 
execution, e.g., all constraints in the standard Kaluza benchmarks
\cite{Berkeley-JavaScript} with 50,000+ test cases
generated by symbolic execution on JavaScript applications were
noted in \cite{Vijay-length} to satisfy this condition. 
Intuitively, as elegantly described in \cite{BTV09}, constraints from symbolic 
execution on string-manipulating programs can be viewed as the problem of path 
feasibility over loopless string-manipulating programs $S$ with variable 
assignments and assertions, i.e., generated by the grammar
\begin{equation*}
    S ::= y := f(x_1,\ldots,x_n) \ |\ \text{\ASSERT{$g(x_1,\ldots,x_n)$}}\ |\ 
            S_1; S_2\ 
\end{equation*}
where $f: (\Sigma^*)^n \to \Sigma^*$ and $g: (\Sigma^*)^n \to \{0,1\}$ are
some string functions.
Straight-line programs with assertions can be obtained by turning such programs 
into a Static Single Assignment (SSA) form (i.e. introduce a new variable 
on the left hand side of each assignment).
\OMIT{
As a matter of fact, it was shown in~\cite{LB16} that any 
incorporating a simple form of the $\replaceall$ function, where the pattern and the replacement are both string constants, into the theory of concatenations already results in an undecidable theory of strings.  
Therefore, it is challenging to reason about the string-replace function in its general form. 
}
A partial decidability result 
can be deduced from \cite{LB16}
for the straight-line fragment of the theory of strings, where (1) $f$ in the 
above
grammar is either a concatenation of string constants and variables, or 
the $\replaceall$ function where \emph{the 
pattern and the replacement are both string constants}, and (2) $g$ is
a boolean combination of regular constraints.
In fact, the decision procedure therein admits finite-state transducers, which
subsume only the aforementioned simple form of the $\replaceall$ function.
The decidability boundary of the straight-line fragment involving the 
$\replaceall$ function in its general form (e.g., 
when the replacement parameter is a variable) remains open.

\paragraph{Contribution.} We investigate the decidability boundary of the theory
$\strline[\replaceall]$ of strings involving 
the $\replaceall$ function and regular constraints, with the straight-line
restriction introduced in \cite{LB16}. We provide a decidability result for a 
large fragment of $\strline[\replaceall]$, which is sufficiently powerful to
express the concatenation operator. We show that this decidability result is in a sense maximal
by showing that several important natural extensions of the logic result in undecidability.
We detail these results below:

\OMIT{
We first show that---as mentioned earlier---with the $\replaceall$ function in its general form, the concatenation operation is in fact \emph{redundant}, in the sense that it can be simulated by the $\replaceall$ function. This motivates us to consider a theory of strings where the $\replaceall$ function, instead of the concatenation operation, and the regular constraints, are the basic modalities. We focus on the straight-line fragment of this theory, denoted by $\strline[\replaceall]$.
}

\begin{itemize}
\item If the pattern parameters of the $\replaceall$ function are allowed to be variables, then the satisfiability of $\strline[\replaceall]$ is undecidable (cf. Proposition~\ref{prop-und-pat-var}).
\item If the pattern parameters of the $\replaceall$ function are regular
    expressions, then the satisfiability of $\strline[\replaceall]$ is decidable
        and in EXPSPACE (cf. Theorem~\ref{thm-main}). In addition, we show that
        the satisfiability problem is PSPACE-complete for several cases that are
        meaningful in practice (cf. Corollary~\ref{cor-pspace}). This strictly
        generalises the decidability result in \cite{LB16} of the straight-line 
        fragment with concatenation, regular constraints, and the $\replaceall$ 
        function where
        patterns/replacement parameters are constant strings.
\item If $\strline[\replaceall]$, where the pattern parameter of the
    $\replaceall$ function is a constant letter, is extended with the 
        string-length constraint, then satisfiability becomes undecidable
        again. In fact, this undecidability can be obtained with
        either integer constraints, character constraints, or constraints 
        involving the $\indexof$ function
        (cf. Theorem~\ref{thm-ext-int} and 
        Proposition~\ref{prop-ext-ch-index}).
\end{itemize}

Our decision procedure for $\strline[\replaceall]$ where the pattern parameters
of the $\replaceall$ function are regular expressions follows an 
automata-theoretic approach. The key idea can be illustrated as follows. Let 
us consider the simple formula $C \equiv x = \replaceall(y, a, z) \wedge x \in e_1 \wedge y \in e_2 \wedge z \in e_3$. 
Suppose that $\cA_1,\cA_2,\cA_3$ are the nondeterministic finite state automata corresponding to $e_1,e_2,e_3$ respectively. 
We effectively eliminate the use of $\replaceall$ by nondeterministically
generating from $\cA_1$ a new regular constraint $\cA'_2$ for $y$ as well as a new regular constraint
$\cA'_3$ for $z$.   These constraints incorporate the effect of the
$\replaceall$ function (i.e.\ all regular constraints are on the ``source'' variables).
Then, the satisfiability of $C$ is turned into testing the nonemptiness of the intersection of $\cA_2$ and $\cA'_2$, as well as the nonemptiness of the intersection of $\cA_3$ and $\cA'_3$. When there are multiple occurrences of the $\replaceall$ function, this process can be iterated. 
Our decision procedure enjoys the following advantages:
\begin{itemize}
	\item It is automata-theoretic and built on clean automaton constructions, 
        moreover, when the formula is satisfiable, a solution can be 
        synthesised. For example, in the aforementioned XSS vulnerability detection example, one can synthesise the values of the variables $user$ and $bio$ for a potential attack. 
	\item The decision procedure is modular in
        that the $\replaceall$ terms are removed one by one to
        generate more and more regular constraints (emptiness of the
        intersection of regular constraints could be efficiently handled by
        state-of-the-art solvers like \cite{fang-yu-circuits}).  
    \item The decision procedure requires exponential space (thus double
        exponential time), but under assumptions that are reasonable in practice, 
        the decision procedure uses only polynomial space, which is not
        worse than other string logics (which can encode the PSPACE-complete
        problem of checking emptiness of the intersection of regular 
        constraints). 
\end{itemize}

%

\paragraph{Organisation.} 
This paper is organised as follows: Preliminaries are given in
Section~\ref{sec-prel}. The core string language is defined in
Section~\ref{sec-core}. The main results of this paper are summarised in
Section~\ref{sec-sat}. The decision procedure is presented in
Section~\ref{sec:replaceallsl}-\ref{sec:replaceallre}, case by case. The
extensions of the core string language are investigated in
Section~\ref{sec-ext}. The related work can be found in Section~\ref{sec-rel}.
The \shortlong{full version}{appendix} contains missing proofs and additional
examples.


\section{Preliminaries}\label{sec-prel}

\paragraph{General Notation} 
Let $\mathbb{Z}$ and $\Nat$ denote the set of integers and natural numbers respectively. For $k \in \Nat$, let $[k] = \{1,\cdots, k\}$. For a vector $\vec{x}=(x_1,\cdots, x_n)$, let $|\vec{x}|$ denote the length of $\vec{x}$ (i.e., $n$) and  $\vec{x}[i]$ denote $x_i$ for each $i \in [n]$. 

\paragraph{Regular Languages}
Fix a finite \emph{alphabet} $\Sigma$. Elements in $\Sigma^*$ are called \emph{strings}. Let $\varepsilon$ denote the empty string and  $\Sigma^+ = \Sigma^* \setminus \{\varepsilon\}$. We will use $a,b,\cdots$ to denote letters from $\Sigma$ and $u, v, w, \cdots$ to denote strings from $\Sigma^*$. For a string $u \in \Sigma^*$, let $|u|$ denote the \emph{length} of $u$ (in particular, $|\varepsilon|=0$). A \emph{position} of a nonempty string $u$ of length $n$ is a number $i \in [n]$ (Note that the first position is $1$, instead of  0). In addition, for $i \in [|u|]$, let $u[i]$ denote the $i$-th letter of $u$. 
For two strings $u_1, u_2$, we use $u_1 \cdot u_2$ to denote the \emph{concatenation} of $u_1$ and $u_2$, that is, the string $v$ such that $|v|= |u_1| + |u_2|$ and for each $i \in [|u_1|]$, $v[i]= u_1[i]$ and for each $i \in |u_2|$, $v[|u_1|+i]=u_2[i]$. Let $u, v$ be two strings. If $v = u \cdot v'$ for some string $v'$, then $u$ is said to be a \emph{prefix} of $v$. In addition, if $u \neq v$, then $u$ is said to be a \emph{strict} prefix of $v$. If $u$ is a prefix of $v$, that is, $v = u \cdot v'$ for some string $v'$, then 
we use $u^{-1} v$ to denote $v'$. In particular, $\varepsilon^{-1} v = v$.

A \emph{language} over $\Sigma$ is a subset of $\Sigma^*$. We will use $L_1, L_2, \dots$ to denote languages. For two languages $L_1, L_2$, we use $L_1 \cup L_2$ to denote the union of $L_1$ and $L_2$, and $L_1 \cdot L_2$ to denote the concatenation of $L_1$ and $L_2$, that is, the language $\{u_1 \cdot u_2 \mid u_1 \in L_1, u_2 \in L_2\}$. For a language $L$ and $n \in \Nat$, we define $L^n$, the \emph{iteration} of $L$ for $n$ times, inductively as follows: $L^0=\{\varepsilon\}$ and $L^{n} =L \cdot L^{n-1}$ for $n > 0$. We also use $L^*$ to denote the iteration of $L$ for arbitrarily many times, that is, $L^* = \bigcup \limits_{n \in \Nat} L^n$. Moreover, let $L^+ = \bigcup \limits_{n \in \Nat \setminus \{0\}} L^n$.

\begin{definition}[Regular expressions $\regexp$]
	\[e \eqdef \emptyset \mid \varepsilon \mid a \mid e + e \mid e \concat e \mid e^*, \mbox{ where } a \in \Sigma. \]
	Since $+$ is associative and commutative, we also write $(e_1 + e_2) + e_3$ as $e_1 + e_2 + e_3$ for brevity. We use the abbreviation $e^+ \equiv e \concat e^*$. Moreover, for $\Gamma = \{a_1, \cdots, a_n\}\subseteq \Sigma$, we use the abbreviations $\Gamma \equiv a_1 + \cdots + a_n$ and $\Gamma^\ast \equiv (a_1 + \cdots + a_n)^\ast$. 
\end{definition}
We define $\Ll(e)$ to be the language defined by $e$, that is, the set of strings that match $e$, inductively as follows: $\Ll(\emptyset) =\emptyset$,
$\Ll(\varepsilon) =\{\varepsilon\}$,
%
$\Ll(a)= \{a\}$,
%
$\Ll(e_1 + e_2) = \Ll(e_1) \cup \Ll(e_2)$,
%
$\Ll(e_1 \concat e_2) = \Ll(e_1) \cdot \Ll(e_2)$,
%
$\Ll(e_1^*)=(\Ll(e_1))^*$.
In addition, we use $|e|$ to denote the number of symbols occurring in $e$.

A \emph{nondeterministic finite automaton} (NFA) $\cA$ on $\Sigma$ is a tuple $(Q, \delta, q_0, F)$, where $Q$ is a finite set of \emph{states}, $q_0 \in Q$ is the \emph{initial} state, $F \subseteq Q$ is the set of \emph{final} states, and $\delta \subseteq Q \times \Sigma \times Q$ is the \emph{transition relation}. For a string $w = a_1 \dots a_n$, a \emph{run} of $\cA$ on $w$ is a state sequence $q_0 \dots q_n$ such that for each $i \in [n]$, $(q_{i-1}, a_i, q_i) \in \delta$. A run $q_0 \dots q_n$ is \emph{accepting} if $q_n \in F$. A string $w$ is \emph{accepted} by $\cA$ if there is an accepting run of $\cA$ on $w$. We use $\Ll(\cA)$ to denote the language defined by $\cA$, that is, the set of strings accepted by $\cA$. We will use $\cA, \cB, \cdots$ to denote NFAs. 
For a string $w= a_1 \dots a_n$, we also use the notation $q_1 \xrightarrow[\cA]{w} q_{n+1}$ to denote the fact that there are $q_2,\dots, q_n \in Q$ such that for each $i \in [n]$, $(q_i, a_i, q_{i+1}) \in \delta$.  For an NFA $\cA=(Q, \delta, q_0, F)$ and $q, q' \in Q$, we use $\cA(q,q')$ to denote the NFA obtained from $\cA$ by changing the initial state to $q$ and the set of final states to $\{q'\}$. The \emph{size} of an NFA $\cA=(Q, \delta, q_0, F)$, denoted by $|\cA|$, is defined as $|Q|$, the number of states. For convenience, we will also call an NFA without initial and final states, that is, a pair $(Q, \delta)$, as a \emph{transition graph}. 

It is well-known (e.g. see \cite{HU79}) that regular expressions and NFAs are 
expressively equivalent, and generate precisely all \emph{regular languages}.
In particular, from a regular expression, an equivalent NFA can be constructed 
in linear time. Moreover, regular languages are closed under Boolean
operations, i.e., union, intersection, and complementation.
In particular, given two NFA $\cA_1=(Q_1, \delta_1, q_{0,1}, F_1)$ and
$\cA_2=(Q_2, \delta_2, q_{0,2}, F_2)$ on $\Sigma$, the intersection $\Ll(\cA_1)
\cap \Ll(\cA_2)$ is recognised by the \emph{product automaton} $\cA_1 \times
\cA_2$ of $\cA_1$ and $\cA_2$ defined as $(Q_1 \times Q_2, \delta, (q_{0,1}, q_{0,2}), F_1 \times F_2)$, where $\delta$ comprises the transitions $((q_1, q_2), a, (q'_1, q'_2))$ such that $(q_1, a, q'_1) \in \delta_1$ and $(q_2, a, q'_2) \in \delta_2$.  

\paragraph{Graph-Theoretical Notation}
A DAG (\emph{directed acyclic graph}) $G$ is a finite directed graph $(V, E)$ with
no directed cycles, where $V$ (resp.~$E \subseteq V \times V$) is a set of vertices (resp.~edges).
Equivalently, a DAG is a directed graph that has a topological ordering, which
is a sequence of the vertices such that every edge is directed from an earlier 
vertex to a later vertex in the sequence. An edge $(\mathit{v},\mathit{v'})$ in
$G$ is called an \emph{incoming} edge of $\mathit{v'}$ and an \emph{outgoing}
edge of $\mathit{v}$. If $(\mathit{v},\mathit{v'}) \in E$, then $\mathit{v'}$ is
called a \emph{successor} of $\mathit{v}$ and $\mathit{v}$ is called a
\emph{predecessor} of $\mathit{v'}$. A \emph{path} $\pi$ in $G$ is a sequence
$\mathit{v}_0 \mathit{e}_1 \mathit{v}_1 \cdots \mathit{v}_{n-1} \mathit{e}_n
\mathit{v}_n$ such that for each $i \in [n]$, we have $\mathit{e}_i =
(\mathit{v}_{i-1},\mathit{v}_i) \in E$. The \emph{length} of the path $\pi$
is the number $n$ of edges in $\pi$. If there is a path from
$\mathit{v}$ to $\mathit{v'}$ (resp. from $\mathit{v'}$ to $\mathit{v}$) in $G$,
then $\mathit{v'}$ is said to be \emph{reachable} (resp. \emph{co-reachable})
from $\mathit{v}$ in $G$. If $\mathit{v}$ is reachable from $\mathit{v'}$ in
$G$, then $\mathit{v'}$ is also called an \emph{ancestor} of $\mathit{v}$ in
$G$. In addition, an edge $(\mathit{v'},\mathit{v''})$ is said to be reachable 
(resp. co-reachable) from $\mathit{v}$ if $\mathit{v'}$ is reachable from $\mathit{v}$ (resp. $\mathit{v''}$ is co-reachable from $\mathit{v}$). The \emph{in-degree} (resp. \emph{out-degree}) of a vertex $\mathit{v}$ is the number of incoming (resp. outgoing) edges of $\mathit{v}$. 
A \emph{subgraph} $G'$ of $G=(V,E)$ is a directed graph $(V', E')$ with
$V' \subseteq V$ and $E' \subseteq E$. Let $G'$ be a subgraph of $G$. Then $G \setminus G'$ is the graph obtained from $G$ by removing all the edges in $G'$. 
\anthony{Is this supposed to be ``removing all edges and vertices of $G'$?}

\hide{
\begin{definition}[Diamond and diamond graph]
Let $G=(V,E)$  be a DAG and $v,v' \in V$ with $v \neq v'$. Then a diamond in $G$ from $v$ to $v'$ is a pair of paths $\pi_1, \pi_2$ from $v$ to $v'$ such that  $\pi_1$ and $\pi_2$ are vertex-disjoint, except $v$ and $v'$. The diamond graph of $G$, denoted by $\cG_{\sf dmd}(G)$, is the graph $(G, E \cup E')$, where $E'$ comprises the pairs $(v,v')$ such that $v\neq v'$ and there is a diamond from $v$ to $v'$. The edges in $E'$ are called the diamond edges of $\cG_{\sf dmd}(G)$.
\end{definition}

\begin{definition}[Diamond index]
Let $G=(V,E)$ be a DAG and $\cG_{\sf dmd}(G)$ be the diamond graph of $G$. Then the diamond index of $G$, denoted by $\dmdidx(G)$, is the maximum number of diamond edges in a path of $\cG_{\sf dmd}(G)$.
\end{definition}

\begin{example}
Example for diamond index
\end{example}

\begin{proposition}\label{prop-num-path}
Let $G=(V,E)$ be a DAG such that the out-degree of each vertex is at most two. Then there are $n^{O(\dmdidx(G))}$ different paths\footnote{two paths are different iff the set of edges in the two paths are different.}  in $G$.
\end{proposition}
}

\paragraph{Computational Complexity}
In this paper, we study not only decidability but also the complexity of string logics. 
In particular, we shall deal with the following computational complexity
classes (see \cite{HU79} for more details): PSPACE (problems solvable in polynomial
space and thus in exponential time), and EXPSPACE (problems solvable
in exponential space and thus in double exponential time). Verification
problems that have complexity PSPACE or beyond (see \cite{BK08}
for a few examples) have substantially benefited from techniques
such as symbolic model checking \cite{McMillan}.


\section{The core constraint language}\label{sec-core}

In this section, we define a general string constraint language that supports 
concatenation, the $\replaceall$ function, and regular constraints. Throughout this section, we fix an alphabet $\Sigma$.

\subsection{Semantics of the $\replaceall$ Function}
To define the semantics of the $\replaceall$ function, we note that the function encompasses three parameters: the first parameter is the \emph{subject} string, the second parameter is a \emph{pattern} that is a string or a regular expression, and the third parameter is the \emph{replacement} string. When the pattern parameter is a string, the semantics is somehow self-explanatory. However, when it is a regular expression, there is no consensus on the semantics even for the mainstream programming languages such as Python and Javascript.
 This is particularly the case when interpreting the union (aka alternation) operator in regular expressions or performing a $\replaceall$ with a pattern that matches $\varepsilon$. In this paper, we mainly 
 focus on the semantics of \emph{leftmost and longest matching}.
Our handling of $\varepsilon$ matches is consistent with our testing of the implementation in Python and the \texttt{sed} command with the \texttt{--posix} flag.
We also assume union is commutative (e.g.\ $\replaceall(aa, a + aa, b) = \replaceall(aa, aa + a, b) = b$) as specified by POSIX, but often ignored in practice (where $bb$ is a common result in the former case).

\hide{
\begin{definition}
Let $u, v$ be two strings such that $v = v_1 u v_2$ for some $v_1,v_2$ and $e$ be a regular expression such that $\varepsilon \not \in \Ll(e)$. We say that $u$ is the \emph{leftmost and longest} matching of $e$ in $v$ if the following two conditions hold
\begin{enumerate}
	\item leftmost: $u \in \Ll(e)$,  and $(v'_1)^{-1} v \not \in  \Ll(e \concat \Sigma^*)$ for every strict prefix $v'_1$ of $v_1$, 
	\item longest: for every nonempty prefix $v'_2$ of $v_2$, $u \cdot v'_2 \not \in \Ll(e)$.
\end{enumerate} 
\end{definition}
}

\begin{definition}
Let $u, v$ be two strings such that $v = v_1 u v_2$ for some $v_1,v_2$ and $e$ be a regular expression. We say that $u$ is the \emph{leftmost and longest} matching of $e$ in $v$ if one of the following two conditions hold,
\begin{itemize}
\item case $\varepsilon \not \in \Ll(e)$:
\begin{enumerate}
	\item leftmost: $u \in \Ll(e)$,  and $(v'_1)^{-1} v \not \in  \Ll(e \concat \Sigma^*)$ for every strict prefix $v'_1$ of $v_1$, 
	\item longest: for every nonempty prefix $v'_2$ of $v_2$, $u \cdot v'_2 \not \in \Ll(e)$.
\end{enumerate} 
\item case $\varepsilon \in \Ll(e)$:
\begin{enumerate}
	\item leftmost: $u \in \Ll(e)$, and $v_1 = \varepsilon$, 
	\item longest: for every nonempty prefix $v'_2$ of $v_2$, $u \cdot v'_2 \not \in \Ll(e)$.
\end{enumerate} 
\end{itemize}
\end{definition}

\begin{example}
Let us first consider $\Sigma = \{0,1\}$, $v=1010101$, $v_1 =1$, $u = 010$, $v_2 = 101$, and $e = 0^*01(0^*+ 1^*)$. Then $v= v_1 u v_2$, and the leftmost and longest matching of $e$ in $v$ is $u$. This is because $u \in \Ll(e)$, $\varepsilon^{-1} v = v \not \in \Ll(e \concat \Sigma^*)$ (notice that $v_1$ has only one strict prefix, i.e. $\varepsilon$), and none of $u 1=0101$, $u 10=01010$, and $u101=010101$ belong to $\Ll(e)$ (notice that $v_2$ has three nonempty prefixes, i.e. $1,10,101$). For another example, let us consider $\Sigma = \{a,b,c\}$, $v=baac$, $v_1 = \varepsilon$, $u =\varepsilon$, $v_2 = v$, and $e = a^*$. Then $v = v_1 u v_2$ and  the leftmost and longest matching of $e$ in $v$ is $u$. This is because $u \in \Ll(e)$, $v_1 = \varepsilon$, and $b, ba, baa, baac \not \in \Ll(e)$. On the other hand, similarly, one can verify that the leftmost and longest matching of $e=a^*$ in $v=aac$ is $u=aa$.
\end{example}

\hide{
\begin{definition} \label{def:replaceall}
The semantics of $\replaceall(u, e, v)$, where $u, v$ are strings and $e$ is a regular expression, is defined inductively as follows:
\begin{itemize}
%
%
	\item if $u \not \in \Ll(\Sigma^\ast \concat e \concat \Sigma^\ast)$, that is, $u$ does \emph{not} contain any substring from $\Ll(e)$, then $\replaceall(u, e, v) = u$, 
	\item otherwise, let $u = u_1 \cdot u_2 \cdot u_3$ such that $u_2$ is the \emph{leftmost and longest} matching of $e$ in $u$, then 
	\begin{itemize}
	\item if $u_2 \neq \varepsilon$, then $\replaceall(u, e, v) = u_1 \cdot v \cdot \replaceall(u_3, e, v)$,
	\item otherwise (in this case, $u_2 = \varepsilon \in \Ll(e)$, thus $u_1 = \varepsilon$ as well), if $u = \varepsilon$, then $\replaceall(u, e, v) = v$, otherwise,  let $u = a \cdot u'$ for some $a \in \Sigma$ and $u' \in \Sigma^*$, then $\replaceall(u, e, v) = v \cdot a \cdot \replaceall(u', e, v)$. 
	\end{itemize}
\end{itemize}
\end{definition}
}

\begin{definition} \label{def:replaceall}
The semantics of $\replaceall(u, e, v)$, where $u, v$ are strings and $e$ is a regular expression, is defined inductively as follows:
\begin{itemize}
	\item if $u \not \in \Ll(\Sigma^\ast \concat e \concat \Sigma^\ast)$, that is, $u$ does \emph{not} contain any substring from $\Ll(e)$, then $\replaceall(u, e, v) = u$, 
	\item otherwise, 
	\begin{itemize}
	\item if $\varepsilon \in \Ll(e)$ and $u$ is the leftmost and longest matching of $e$ in $u$, then $\replaceall(u, e, v) = v$,
	\item if $\varepsilon \in \Ll(e)$, $u = u_1 \cdot  a \cdot  u_2$, $u_1$ is the leftmost and longest matching of $e$ in $u$, and $a \in \Sigma$, then $\replaceall(u, e, v) = v \cdot a \cdot \replaceall(u_2, e, v)$,
	\item if $\varepsilon \not \in \Ll(e)$, $u = u_1\cdot  u_2\cdot  u_3$, and $u_2$ is the leftmost and longest matching of $e$ in $u$,  then $\replaceall(u, e, v) = u_1 \cdot v \cdot \replaceall(u_3, e, v)$. 
	\end{itemize}
\end{itemize}
\end{definition}

\begin{example}
	
At first, $\replaceall(abab, ab, d) =d \cdot \replaceall(ab, ab, d)= dd \cdot \replaceall(\epsilon, ab, d)=dd \cdot \varepsilon = dd$ and $\replaceall(baac, a^+, b)=bbc$. In addition, $\replaceall(aaaa, ``", d) =dadadadad$ and $\replaceall(baac, a^*, b) =bbbcb$.  The argument for $\replaceall(baac, a^*, b) = bbbcb$ proceeds as follows: The leftmost and longest matching of $a^*$ in $baac$ is $u_1 =\varepsilon$, where $baac = u_1 \cdot b \cdot u_2$ and $u_2 = aac$. Then $\replaceall(baac, a^*, b) = b \cdot b \cdot \replaceall(aac, a^*, b)$. Since $aa$ is the leftmost and longest matching of $a^*$ in $aac$, we have $\replaceall(aac, a^*, b) = b \cdot c \cdot \replaceall(\varepsilon, a^*, b) = bcb$. Therefore, we get $\replaceall(baac, a^*, b) = bbbcb$. (The readers are invited to test this in Python and \texttt{sed}.)

\end{example}


\subsection{Straight-Line String Constraints With the $\replaceall$ Function}

We consider the String data type $\str$, and assume a countable set of variables
$x, y, z, \cdots$ of $\str$.

\begin{definition}[Relational and regular constraints]
	Relational constraints and regular constraints are defined by the following rules,
	\[
	\begin{array}{r c l cr}
	s &\eqdef & x \mid u & \ \ & \mbox{(string terms)}\\
	p &\eqdef & x \mid e & \ \ & \mbox{(pattern terms)}\\
	\varphi &\eqdef & x = s \concat s  \mid  x = \replaceall(s, p, s) \mid \varphi \wedge \varphi & \ \ & \mbox{(relational constraints)}\\
	\psi & \eqdef & x \in e \mid \psi \wedge \psi 
	& \ \ & \mbox{(regular constraints)} \\
	\end{array}
	\]
	where $x$ is a string variable, $u \in \Sigma^\ast$ and $e$ is a regular expression over $\Sigma$. 

\end{definition}

For a formula $\varphi$ (resp. $\psi$), let $\vars(\varphi)$ (resp. $\vars(\psi)$) denote the set of variables occurring in $\varphi$ (resp. $\psi$). Given a relational constraint $\varphi$, a variable $x$ is called a \emph{source variable} of $\varphi$ if $\varphi$ \emph{does not} contain a conjunct of the form $x = s_1 \concat s_2$ or $x = \replaceall(-,-,-)$.

We then notice that, with the $\replaceall$ function in its general form, the concatenation operation is in fact redundant.

\begin{proposition}\label{prop-concat}
	The concatenation  operation ($\concat$) can be simulated  by the $\replaceall$ function.
\end{proposition}
\begin{proof}
	It is sufficient to observe that 
	a relational constraint $x = s_1 \concat s_2$ can be rewritten as
	\[x' = \replaceall(ab, a, s_1) \wedge x = \replaceall(x', b, s_2),\] where $a,b$ are two fresh letters.
\end{proof}

In light of Proposition~\ref{prop-concat}, in the sequel, we will \emph{dispense the concatenation operator} mostly and focus on \textbf{the string constraints that involve  the $\replaceall$ function only}.

Another example to show the power of the $\replaceall$ function is that it can simulate the extension of regular expressions with string variables, which is  supported by the mainstream scripting languages like Python, Javascript, and PHP. For instance, $x \in y^*$ can be expressed by $x =\replaceall(x', a, y) \wedge x' \in a^*$, where $x'$ is a fresh variable and $a$ is a fresh letter.

\medskip

The generality of the constraint language makes it undecidable,
even in very simple cases. To retain decidability, we follow \cite{LB16} and focus on the ``straight-line fragment" of the language. This straight-line fragment captures the structure of straight-line string-manipulating
programs with the $\replaceall$ string operation.  

\begin{definition}[Straight-line relational constraints]
	A relational constraint $ \varphi$ with the $\replaceall$ function is straight-line, if $\varphi \eqdef \bigwedge \limits_{1 \le i \le m} x_i = P_i$ such that
	\begin{itemize}
		\item $x_1,\dots, x_m$ are mutually distinct,
		\item for each $i \in [m]$, all the variables in $P_i$ are either source variables, or variables from $\{x_1,\dots, x_{i-1}\}$,
	\end{itemize}
\end{definition}

\begin{remark}
Checking whether a relational constraint $\varphi$ is straight-line can be done in linear time. 
\end{remark}

\begin{definition}[Straight-line string constraints]
	A straight-line string constraint $C$ with the $\replaceall$ function (denoted by $\strline[\replaceall]$)  is defined as $ \varphi \wedge \psi$,  where 
	\begin{itemize}
		\item $\varphi$ is a straight-line relational constraint with the $\replaceall$ function,  and
		\item $\psi$ is a regular constraint.
	\end{itemize}
\end{definition}

\begin{example}
The following string constraint belongs to $\strline[\replaceall]$: 
$$C \equiv x_2 = \replaceall(x_1, 0, y_1) \wedge x_3 = \replaceall(x_2, 1, y_2) \wedge x_1 \in \{0,1\}^* \wedge y_1 \in 1^* \wedge y_2 \in 0^*.$$
\end{example}
 

\section{The satisfiability problem} \label{sec-sat}
In this paper, we focus on the satisfiability problem of $\strline[\replaceall]$, which is formalised as follows. 

\smallskip

\begin{quote} \centering
\framebox{Given an $\strline[\replaceall]$ constraint $C$, decide whether $C$ is satisfiable.}
\end{quote}
\smallskip

To approach this problem, we identify several fragments of  $\strline[\replaceall]$, depending on whether the pattern and the replacement parameters are constants or variables.  We shall investigate extensively the satisfiability problem of the fragments of $\strline[\replaceall]$. 

%
%
%
%

We begin with the case where the pattern parameters of the $\replaceall$ terms are variables. It turns out that in this case the satisfiability problem of $\strline[\replaceall]$ is undecidable.
The proof is by a reduction from Post's Correspondence Problem. Due to space
constraints we relegate the proof to \shortlong{the full
version}{Appendix~\ref{sec:prop-und-pat-var-proof}}.
%

\begin{proposition}\label{prop-und-pat-var}
The satisfiability problem of $\strline[\replaceall]$ is undecidable, if the pattern parameters of the $\replaceall$ terms are allowed to be variables.
\end{proposition}

In light of Proposition~\ref{prop-und-pat-var}, we shall focus on the case that the pattern parameters of the $\replaceall$ terms are constants, being a single letter, a constant string, or a regular expression. 
The main result of the paper is summarised as the following Theorem~\ref{thm-main}.

\begin{theorem}\label{thm-main}
	The satisfiability problem of $\strline[\replaceall]$ is decidable in EXPSPACE, if the pattern parameters of the $\replaceall$ terms are regular expressions.  
\end{theorem}

The following three sections are devoted to the proof of Theorem~\ref{thm-main}.  
\begin{itemize}
\item We start with the \emph{single-letter} case that the pattern parameters of the $\replaceall$ terms are single letters (Section~\ref{sec:replaceallsl}),
\item then consider the \emph{constant-string} case that the pattern parameters of the $\replaceall$ terms are constant strings  (Section~\ref{sec:replaceallcs}), \item and finally the \emph{regular-expression} case that the pattern parameters of the $\replaceall$ terms are regular expressions  (Section~\ref{sec:replaceallre}).
\end{itemize}



We first introduce a graphical representation of $\strline[\replaceall]$ formulae as follows.    

\begin{definition}[Dependency graph]
\label{def:dep-graph}
	Suppose $C= \varphi \wedge \psi$ is an $\strline[\replaceall]$ formula where the pattern parameters of the $\replaceall$ terms are regular expressions. 
Define the \emph{dependency graph} of $C$ as $G_C= (\vars(\varphi), E_C)$, such that for each $i \in [m]$, if $x_i = \replaceall(z, e_i, z')$, then $(x_i, (\rpleft, e_i), z) \in E_C$ and $(x_i, (\rpright, e_i), z') \in E_C$. A final (resp. initial) vertex in $G_C$ is a vertex in $G_C$ without successors (resp. predecessors). The edges labelled by $(\rpleft, e_i)$ and $(\rpright, e_i)$ are called the $\rpleft$-edges and $\rpright$-edges respectively. The \emph{depth} of $G_C$ is the maximum length of the paths in $G_C$. In particular, if $\varphi$ is empty, then the depth of $G_C$ is zero. 
\end{definition}
Note that $G_C$ is a DAG where the out-degree of each vertex is two or zero. 

\begin{definition}[Diamond index and $\rpleft$-length]
Let $C$  be an  $\strline[\replaceall]$ formula and $G_C=(\vars(\varphi), E_C)$ be its dependency graph. A \emph{diamond} $\Delta$ in $G_C$ is a pair of vertex-disjoint simple paths from $z$ to $z'$ for some $z,z' \in \vars(\varphi)$. The vertices $z$ and $z'$ are called the \emph{source} and \emph{destination} vertex of the diamond respectively. A diamond $\Delta_2$ with the source vertex $z_2$ and destination vertex $z'_2$ is said to be reachable from  another diamond $\Delta_1$ with the source vertex $z_1$ and destination vertex $z'_1$ if $z_2$ is reachable from $z'_1$ (possibly $z_2= z'_1$). The \emph{diamond index} of $G_C$, denoted by $\dmdidx(G_C)$, is defined as the maximum length of the diamond sequences $\Delta_1 \cdots \Delta_n$ in $G_C$ such that for each $i \in [n-1]$, $\Delta_{i+1}$ is reachable from $\Delta_i$. The \emph{$\rpleft$-length} of a path in $G_C$ is the number of $\rpleft$-edges in the path. The $\rpleft$-length of $G_C$, denoted by $\lftlen(G_C)$, is the maximum $\rpleft$-length of paths in $G_C$.
\end{definition}
For each dependency graph $G_C$, since each diamond uses at least one $\rpleft$-edge, we know that $\dmdidx(G_C) \le \lftlen(G_C)$.

\begin{proposition}\label{prop-di}
Let $C$  be an  $\strline[\replaceall]$ formula and $G_C=(\vars(\varphi), E_C)$ be its dependency graph. For each pair of distinct vertices $z,z'$ in $G_C$, there are at most $(|\vars(\varphi)||E_C|)^{O(\dmdidx(G_C))}$ different paths from $z$ to $z'$.
\end{proposition}
It follows from Proposition~\ref{prop-di} that for a class of $\strline[\replaceall]$ formulae $C$ such that $\dmdidx(G_C)$ is bounded by a constant $c$,  there are polynomially many different paths between each pair of distinct vertices in $G_C$.

\begin{example}
Let $G_C$ be the dependency graph illustrated in Figure~\ref{fig-dmdidx-exmp}. It is easy to see that $\dmdidx(G_C)$ is $3$. In addition, there are $2^3=8$ paths from $x_1$ to $y_1$. If we generalise $G_C$ in Figure~\ref{fig-dmdidx-exmp} to a dependency graph comprising $n$ diamonds from $x_1$ to $x_2$, $\cdots$, from $x_{n-1}$ to $x_n$, and from $x_n$ to $y_1$ respectively, then the diamond index of the resulting dependency graph is $n$ and there are $2^n$ paths from $x_1$ to $y_1$ in the graph.
\begin{figure}[htbp]
\begin{center}
\includegraphics[scale=0.7]{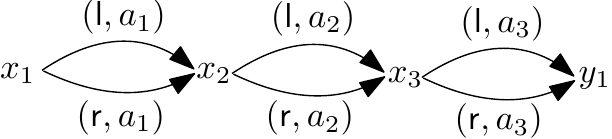}
\end{center}
\caption{The diamond index  and the number of paths in $G_C$}\label{fig-dmdidx-exmp}
\end{figure}
\end{example}

In Section~\ref{sec:replaceallsl}--\ref{sec:replaceallre}, we will apply a refined analysis of the complexity of the decision procedures for proving Theorem~\ref{thm-main} and get the following results.

\begin{corollary}\label{cor-pspace}
The satisfiability problem is PSPACE-complete for the following fragments of $\strline[\replaceall]$:
\begin{itemize}
\item the single-letter case, plus the condition that the diamond indices of the dependency graphs are bounded by a constant $c$, 
\item the constant-string case, plus the condition that the $\rpleft$-lengths of the dependency graphs are bounded by a constant $c$, 

\item the regular-expression case, plus the condition that the $\rpleft$-lengths of the dependency graphs are at most $1$.
\end{itemize}
\end{corollary}

Corollary~\ref{cor-pspace} partially justifies our choice to present the decision procedures for the single-letter, constant-string, and regular-expression case separately. Intuitively, when the pattern parameters  of the $\replaceall$ terms become less restrictive, the decision procedures  become more involved, and more constraints should be imposed on the dependency graphs in order to achieve the PSPACE upper-bound. The PSPACE lower-bound follows from the observation that  nonemptiness of the intersection of the regular expressions $e_1, \cdots, e_n$ over the alphabet $\{0,1\}$, which is a PSPACE-complete problem, can be reduced to the satisfiability of the formula $x \in e_1 \wedge \cdots \wedge x \in e_n$, 
%
which falls into all fragments of $\strline[\replaceall]$ specified in Corollary~\ref{cor-pspace}.\zhilin{Remark for the pspace lower bound, please check.}\tl{seems to be fine.I edited a little bit}
At last, we  remark that the restrictions in Corollary~\ref{cor-pspace} are partially inspired by the  benchmarks in practice. Diamond indices (intuitively, the ``nesting depth'' of $\replaceall(x,a,x)$) are likely to be small in practice because the constraints like $\replaceall(x,a,x)$ are rather artificial and rarely occur in practice. Moreover, the $l$-length reflects the nesting depth of replaceall w.r.t. the first parameter, which is also likely to be small. Finally, for string constraints with concatenation and $\replaceall$ where pattern/replacement parameters are constants, the diamond index is no greater than the ``dimension'' defined in \cite{LB16}, where it was shown that existing benchmarks mostly have ``dimensions" at most three for such string constraints. \zhilin{In the rebuttal, we stated that diamond index and dimension for constraints which use only concatenation. But I think the statement is wrong. Please check.}\tl{the current version is fine. I edited a little bit.}



\section{Outline of Decision Procedures}

We describe our decision procedure across three sections (Section~\ref{sec:replaceallsl}--Section~\ref{sec:replaceallre}).
This means the ideas can be introduced in a step-by-step fashion, which we hope helps the reader.
In addition, by presenting separate algorithms, we can give the fine-grained complexity analysis required to show Corollary~\ref{cor-pspace}.
We first outline the main ideas needed by our approach.

We will use automata-theoretic techniques.
That is, we make use of the fact that regular expressions can be represented as NFAs.
We can then consider a very simple string expression, which is a single regular constraint $x \in e$.
It is well-known that an NFA $\cA$ can be constructed that is equivalent to $e$.
We can also test in LOGSPACE whether there is some word $w$ accepted by $\cA$.
If this is the case, then this word can be assigned to $x$, giving a satisfying assignment to the constraint.
If this is not the case, then there is no satisfying assignment.

A more complex case is a conjunction of several constraints of the form $x \in e$.
If the constraints apply to different variables, they can be treated independently to find satisfying assignments.
If the constraints apply to the same variable, then they can be merged into a single NFA.\@
Intuitively, take $x \in e_1 \land x \in e_2$ and $\cA_1$ and $\cA_2$ equivalent to $e_1$ and $e_2$ respectively.
We can use the fact that NFA are closed under intersection a check if there is a word accepted by $\cA_1 \times \cA_2$.
If this is the case, we can construct a satisfying assignment to $x$ from an accepting run of $\cA_1 \times \cA_2$.

In the general case, however, variables are not independent, but may be related by a use of $\replaceall$.
In this case, we perform a kind of \emph{$\replaceall$ elimination}.
That is, we successively remove instances of $\replaceall$ from the constraint, building up an expanded set of regular constraints (represented as automata).
Once there are no more instances of $\replaceall$ we can solve the regular constraints as above.
Briefly, we identify some $x = \replaceall(y, e, z)$ where $x$ does not appear as an argument to any other use of $\replaceall$.
We then transform any regular constraints on $x$ into additional constraints on $y$ and $z$.
This allows us to remove the variable $x$ since the extended constraints on $y$ and $z$ are sufficient for determining satisfiability.
Moreover, from a satisfying assignment to $y$ and $z$ we can construct a satisfying assignment to $x$ as well.
This is the technical part of our decision procedure and is explained in detail in the following sections, for increasingly complex uses of $\replaceall$.


\section{Decision procedure for $\strline[\replaceall]$: The single-letter case} \label{sec:replaceallsl}

\hide{
\subsection{Encode Concatenation by ReplaceAll}

We will transform a given constraint $\varphi$ into a constraint $\varphi'$ which is concatenation free. As the first step, we extend the original alphabet with two fresh letters $a,b$.

For each $x=yz$, we introduce a new variable $x'$ and replace $x=yz$ by two new constraints
$x'=\replaceall(ab, a, x)$ and $x=\replaceall(x', b, z)$.

\begin{proposition}
	$\varphi$ and $\varphi'$ are equisatisfiable.
\end{proposition}
}


In this section, we consider the single-letter case, that is, for the $\strline[\replaceall]$ formula $C = \varphi \wedge \psi$, every term of the form $\replaceall(z, e, z')$ in $\varphi$ satisfies that $e=a$ for $a \in \Sigma$.
We begin by explaining the idea of the decision procedure in the case where there is a single use of a $\replaceall(-, - ,-)$ term.
Then we describe the decision procedure in full details.



\subsection{A Single Use of $\replaceall(-, -, -)$}
\label{sec:single-use}

Let us start with the simple case that
\[C \equiv x = \replaceall(y, a, z) \wedge x \in e_1 \wedge y \in e_2 \wedge z \in e_3,\]
where, for $i =1, 2, 3$, we suppose  $\cA_i = (Q_i, \delta_i, q_{0,i}, F_i)$
is the NFA corresponding to the regular expression $e_i$.

From the semantics, $C$ is satisfiable if and only if $x, y, z$ can be assigned with strings $u, v, w$ so that: (1) $u$ is obtained from $v$ by replacing all the occurrences of $a$ in $v$ with $w$, and (2) $u, v, w$ are accepted by $\cA_1, \cA_2, \cA_3$ respectively. Let $u,v,w$ be the strings satisfying these two constraints. As $u$ is accepted by $\cA_1$,  there must be an accepting run of $\cA_1$ on $u$. Let $v = v_1 a v_2 a \cdots a v_k$ such that for each $i \in [k]$, $v_i \in (\Sigma \setminus \{a\})^*$. Then $u = v_1 w v_2 w \cdots w v_k$ and there are states $q_1, q'_1, \cdots, q_{k-1}, q'_{k-1}, q_k$  such that
$$
q_{0,1} \xrightarrow[\cA_1]{v_1} q_1 \xrightarrow[\cA_1]{w} q'_1 \xrightarrow[\cA_1]{v_2} q_2 \xrightarrow[\cA_1]{w} q'_2 \cdots q_{k-1} \xrightarrow[\cA_1]{w} q'_{k-1} \xrightarrow[\cA_1]{v_k} q_k
$$
 and $q_k \in F_{1}$. Let $T_z$ denote $\left\{(q_i, q'_i) \mid i \in [k-1] \right\}$. Then $w \in \Ll(\cA_3)\ \cap\ \bigcap \limits_{(q, q') \in T_z} \Ll(\cA_1(q, q'))$. In addition, let  $\cB_{\cA_1,a,T_z}$ be the NFA obtained from $\cA_1$ by removing all the $a$-transitions first and then adding the $a$-transitions $(q, a, q')$ for $(q, q') \in T_z$. Then
$$
q_{0,1} \xrightarrow[\cB_{\cA_1,a,T_z}]{v_1} q_1 \xrightarrow[\cB_{\cA_1,a,T_z}]{a} q'_1 \xrightarrow[\cB_{\cA_1,a,T_z}]{v_2} q_2 \xrightarrow[\cB_{\cA_1,a,T_z}]{a} q'_2 \cdots q_{k-1} \xrightarrow[\cB_{\cA_1,a,T_z}]{a} q'_{k-1} \xrightarrow[\cB_{\cA_1,a,T_z}]{v_k} q_k.
$$
Therefore,
$v \in \Ll(\cA_2) \cap \Ll(\cB_{\cA_1,a,T_z})$. We deduce that there is $T_z \subseteq Q_1 \times Q_1$ such that $\Ll(\cA_3)\ \cap\ \bigcap \limits_{(q, q') \in T_z} \Ll(\cA_1(q, q')) \neq \emptyset$ and $ \Ll(\cA_2) \cap \Ll(\cB_{\cA_1,a,T_z}) \neq \emptyset$. In addition, it is not hard to see that this condition is also sufficient for the satisfiability of $C$. The arguments proceed as follows: Let $v \in  \Ll(\cA_2) \cap \Ll(\cB_{\cA_1,a,T_z})$ and $w \in \Ll(\cA_3)\ \cap\ \bigcap \limits_{(q, q') \in T_z} \Ll(\cA_1(q, q'))$. From $v \in \Ll(\cB_{\cA_1,a,T_z})$, we know that there is an accepting run of $\cB_{\cA_1,a,T_z}$ on $v$. Recall that $\cB_{\cA_1,a,T_z}$ is obtained from $\cA_1$ by first removing all the $a$-transitions, then adding all the transitions $(q,a,q')$ for $(q,q') \in T_z$.  Suppose $v = v_1 a v_2 \cdots a v_k$ such that $v_i \in (\Sigma \setminus \{a\})^*$ for each $i \in [k]$ and
$$
q_{0,1} \xrightarrow[\cB_{\cA_1,a,T_z}]{v_1} q_1 \xrightarrow[\cB_{\cA_1,a,T_z}]{a} q'_1 \xrightarrow[\cB_{\cA_1,a,T_z}]{v_2} q_2 \xrightarrow[\cB_{\cA_1,a,T_z}]{a} q'_2 \cdots q_{k-1} \xrightarrow[\cB_{\cA_1,a,T_z}]{a} q'_{k-1} \xrightarrow[\cB_{\cA_1,a,T_z}]{v_k} q_k
$$
is an accepting run of $\cB_{\cA_1,a,T_z}$ on $v$. Then $q_{0,1} \xrightarrow[\cA_1]{v_1} q_1$, and for each $i \in [k-1]$ we have $(q_i, q'_i) \in T_z$ and $q'_i \xrightarrow[\cA_1]{v_{i+1}} q_{i+1}$; moreover, $q_k \in F_1$.
Let $u = \replaceall(v, a, w)=v_1 w v_2 \cdots w v_k$. Since $w \in \bigcap \limits_{(q, q') \in T_z} \Ll(\cA_1(q, q'))$,  we infer that
$$
q_{0,1} \xrightarrow[\cA_1]{v_1} q_1 \xrightarrow[\cA_1]{w} q'_1 \xrightarrow[\cA_1]{v_2} q_2 \xrightarrow[\cA_1]{w} q'_2 \cdots q_{k-1} \xrightarrow[\cA_1]{w} q'_{k-1} \xrightarrow[\cA_1]{v_k} q_k
$$
is an accepting run of $\cA_1$ on $u$. Therefore, $u$ is accepted by $\cA_1$ and $C$ is satisfiable.

\begin{proposition}\label{prop-sat-sl-case}
We have $C \equiv x = \replaceall(y, a, z) \wedge x \in e_1 \wedge y \in e_2 \wedge z \in e_3$ is satisfiable iff there exists $T_{z} \subseteq Q_1 \times Q_1$ with $\Ll(\cA_3)\ \cap\ \bigcap \limits_{(q, q') \in T_z} \Ll(\cA_1(q, q')) \neq \emptyset$ and $ \Ll(\cA_2) \cap \Ll(\cB_{\cA_1,a,T_z}) \neq \emptyset$.
\end{proposition}

From Proposition~\ref{prop-sat-sl-case}, we can decide the satisfiability of $C$ in polynomial space as follows:
\begin{description}
\item[Step I.] Nondeterministically choose a set $T_{z} \subseteq Q_1 \times Q_1$.
\item[Step II.] Nondeterministically choose an accepting run of the product automaton of $\cA_3$ and $\cA_1(q, q')$ for $(q,q') \in T_{z}$.
\item[Step III.] Nondeterministically choose an accepting run of the product automaton of $\cA_2$ and $\cB_{\cA_1, a,  T_{z}}$.
\end{description}
During Step II and III, it is sufficient to record $T_z$ and a state of the product automaton, which occupies only a polynomial space.

The above decision procedure can be easily generalised to the case that there are multiple atomic regular constraints for $x$. For instance, let $x \in e_{1,1} \wedge x \in e_{1,2}$ and for $j = 1, 2$, $\cA_{1,j} = (Q_{1,j}, \delta_{1, j}, q_{0,1, j}, F_{1,j})$ be
the NFA corresponding to $e_{1,j}$. Then in Step I, two sets $T_{1,z} \subseteq Q_{1,1} \times Q_{1,1}$ and $T_{2,z} \subseteq Q_{1,2} \times Q_{1,2}$ are nondeterministically chosen, moreover, Step II and III are adjusted accordingly.

\begin{example}\label{exmp-sl}
Let $C \equiv x= \replaceall(y, 0, z) \wedge x \in e_1 \wedge y \in e_2 \wedge z \in e_3$, where $e_1=(0+1)^* (00(0+1)^* + 11(0+1)^*)$,  $e_2= (01)^*$,  and $e_3 = (10)^*$. The NFA $\cA_{1}, \cA_{2}, \cA_{3}$ corresponding to $e_1, e_2, e_3$ respectively are illustrated in Figure~\ref{fig-sl-exmp}. Let $T_z = \{(q_0, q_0), (q_1, q_2)\}$.  Then
$$
\begin{array}{l c l}
 \Ll(\cA_3)\ \cap\ \bigcap \limits_{(q, q') \in T_z} \Ll(\cA_1(q, q'))  & = & \Ll(\cA_3)\ \cap \Ll(\cA_1(q_0, q_0)) \cap \Ll(\cA_1(q_1,q_2)) \\
& = & \Ll((10)^*) \cap \Ll((0+1)^*) \cap \Ll(1(0+1)^*) \\
& \neq & \emptyset.
\end{array}
$$
In addition, $\cB_{\cA_1, 0, T_z}$ (also illustrated in Figure~\ref{fig-sl-exmp}) is obtained from $\cA_1$ by removing all the $0$-transitions, then adding the transitions $(q_0, 0, q_0)$ and $(q_1, 0, q_2)$. Then
$$
 \Ll(\cA_2) \cap \Ll(\cB_{\cA_1,0,T_z})   =  \Ll((01)^*) \cap \Ll((0+1)^* 10 1^*) \neq \emptyset.
$$
We can choose $z$ to be a string from $\Ll(\cA_3) \cap \ \bigcap \limits_{(q, q') \in T_z} \Ll(\cA_1(q, q'))=\Ll((10)^*) \cap \Ll((0+1)^*) \cap \Ll(1(0+1)^*)$, say $10$, and $y$ to be a string from $ \Ll(\cA_2) \cap \Ll(\cB_{\cA_1,0,T_z})   =  \Ll((01)^*) \cap \Ll((0+1)^* 10 1^*)$, say $0101$,  then we set $x$ to $\replaceall(0101, 0, 10)=101101$, which is in $\Ll(\cA_1)$. Thus, $C$ is satisfiable.\qed
\begin{figure}[htbp]
\begin{center}
\includegraphics[scale=0.63]{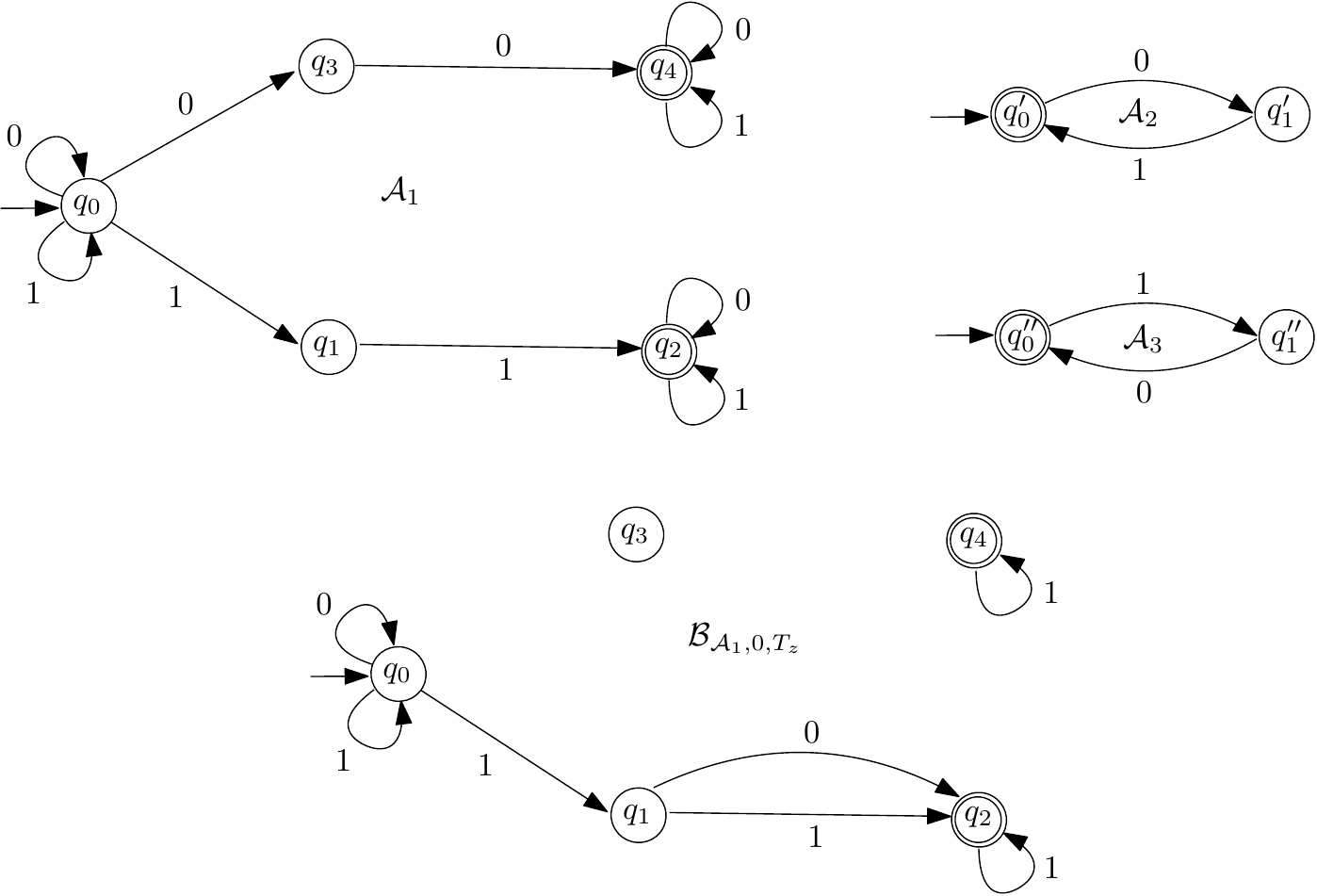}
\end{center}
\caption{An example for the single-letter case: One $\replaceall$}\label{fig-sl-exmp}
\end{figure}
\end{example}

\subsection{The General Case}
\label{sec:dp-sl-general}

Let us now consider the general case where $C$ contains multiple occurrences of $\replaceall(-, -, -)$ terms.
Then the satisfiability of $C$ is decided by the following two-step procedure.

\smallskip

\noindent{\bf Step I.} We utilise the dependency graph $C$ and compute nondeterministically a collection of atomic regular constraints $\cE(x)$ for each variable $x$, in a top-down manner.

Notice that $\cE(x)$ is represented succinctly as a set of pairs $(\cT, \cP)$, where $\cT=(Q, \delta)$ is a transition graph and $\cP \subseteq Q \times Q$. The intention of $(\cT, \cP)$ is to represent succinctly the collection of the atomic regular constraints containing $(Q, \delta, q, \{q'\})$ for each $(q,q') \in \cP$, where $q$ is the initial state and $\{q'\}$ is the set of final states.


Initially, let $G_0:= G_C$. In addition, for each variable $x$, we define $\cE_0(x)$ as follows: Let $x \in e_1 \wedge \cdots \wedge x \in e_n$ be the conjunction of all the atomic regular constraints related to $x$ in $C$. For each $i \in [n]$, let $\cA_i=(Q_i, \delta_i, q_{0,i}, F_i)$ be the NFA corresponding to $e_i$. We nondeterministically choose $q_i \in F_i$ and set $\cE_0(x):=  \left\{((Q_i, \delta_i), \{(q_{0,i}, q_i)\}) \mid i \in [n] \right\}$.

%

We begin with $i:= 0$ and repeat the following procedure until we reach some $i$ where $G_i$ is an empty graph, i.e. a graph without edges.
Note that $G_0$ was defined above.
\begin{enumerate}
\item Select a vertex $x$ of $G_i$ such that $x$ has no predecessors  and has two successors via edges $(x, (\rpleft, a), y)$ and $(x, (\rpright,a), z)$ in $G_i$.  Suppose $\cE_i(x)=\{(\cT_1,\cP_1), \cdots, (\cT_k, \cP_k)\}$, where for each $j \in [k]$, $\cT_j = (Q_j, \delta_j)$. Then $\cE_{i+1}(z)$ and  $\cE_{i+1}(y)$ and $G_{i+1}$ are computed as follows:
\begin{enumerate}
\item For each $j \in [k]$, nondeterministically choose a set $T_{j, z} \subseteq Q_j \times Q_j$.
\item If $y \neq z$, then let
\[
    \qquad\qquad
    \cE_{i+1}(z):= \cE_{i}(z) \cup \left\{(\cT_j, T_{j,z}) \mid j \in [k] \right\}
    \text{\ \ and\ \ }
    \cE_{i+1}(y): = \cE_{i}(y) \cup \left\{(\cT_{\cT_j, a, T_{j,z}}, \cP_j) \mid j \in [k] \right\}
\]
where $\cT_{\cT_j, a, T_{j,z}}$ is obtained from $\cT_j$ by first removing all the $a$-transitions, then adding all the transitions $(q, a, q')$ for $(q,q') \in T_{j,z}$.
Otherwise, let
$\cE_{i+1}(z):= \cE_{i}(z) \cup \left\{(\cT_j, T_{j,z}) \mid j \in [k] \right\} \cup \left\{(\cT_{\cT_j, a, T_{j,z}}, \cP_j) \mid j \in [k] \right\}$.
In addition, for each vertex $x'$ distinct from $y, z$, let $\cE_{i+1}(x') := \cE_i(x')$.
%
\item Let $G_{i+1}:= G_i \setminus \{(x, (\rpleft, a), y), (x, (\rpright,a), z)\}$.
\end{enumerate}
\item Let $i: = i+1$.
\end{enumerate}

For each variable $x$, let $\cE(x)$ denote the set $\cE_i(x)$ after exiting the above loop.

\smallskip

\noindent {\bf Step II.}
Output ``satisfiable'' if for each source variable $x$ there is an accepting run of the product of all the NFA in $\cE(x)$; otherwise, output ``unsatisfiable''.


\smallskip


It remains to argue the correctness and complexity of the above procedure and show how to obtain satisfying assignments to satisfiable constraints.
Correctness follows a similar argument to Proposition~\ref{prop-sat-sl-case} and
is presented in \shortlong{the full
version}{Appendix~\ref{sec:dp-sl-correctness}}.
Intuitively, Proposition~\ref{prop-sat-sl-case} shows our procedure correctly eliminates occurrences of $\replaceall$ until only regular constraints remain.

\mat{Describing model extraction below.}\zhilin{The description is nice. I think we should remove the reference to the appendix.}\tl{zhilin, which refernce?}
If, in the case that the equation \tl{constraint?} is satisfiable, one wishes to obtain a satisfying assignment to all variables, we can proceed as follows.
First, for each source variable $x$, nondeterministically choose an accepting run of the product of all the NFA in $\cE(x)$.
As argued in \shortlong{the full version}{Appendix~\ref{sec:dp-sl-correctness}}, the word labelling this run satisfies all regular constraints on $x$ since it is taken from a language that is guaranteed to be a subset of the set of words satisfying the original constraints.
For non-source variables, we derive an assignment as in Proposition~\ref{prop-sat-sl-case}, proceeding by induction from the source variables.
That is, select some variable $x$ such that $x$ is derived from variables $y$ and $z$ and assignments to both $y$ and $z$ have already been obtained.
The value for $x$ is immediately obtained by performing the $\replaceall$ operation using the assignments to $y$ and $z$.
That this value satisfies all regular constraints on $x$ follows the same argument as Proposition~\ref{prop-sat-sl-case}.
The procedure terminates when all variables have been assigned.

We now give an example before proceeding to the complexity analysis.

\begin{example}
Suppose $C \equiv x= \replaceall(y, 0, z) \wedge y = \replaceall(y', 1, z') \wedge x \in e_1 \wedge y \in e_2 \wedge z \in e_3 \wedge y' \in e_4  \wedge z' \in e_5$, where $e_1, e_2, e_3$ are as in Example~\ref{exmp-sl}, $e_4=0^* 1^* 0^* 1^*$, and $e_5=0^*1^*$. Let $\cA_4,\cA_5$ be the NFA corresponding to $e_4$ and $e_5$ respectively (see Figure~\ref{fig-sl-exmp-nested}). The dependency graph $G_C$ of $C$ is illustrated in Figure~\ref{fig-sl-exmp-nested}. Let $\cT_1,\cdots, \cT_5$ be the transition graph of $\cA_1,\cdots, \cA_5$ respectively. Then the collection of regular constraints $\cE(\cdot)$ are computed as follows.
\begin{itemize}
\item Let $G_0=G_C$. Pick the sets $\cE_0(x) = \{(\cT_1, \{(q_0, q_2)\})\}$, $\cE_0(y) = \{(\cT_2, \{(q^\prime_0, q^\prime_0)\})\}$, $\cE_0(z) = \{(\cT_3, \{(q^{\prime\prime}_0, q^{\prime\prime}_0)\})\}$, $\cE_0(y^\prime) = \{(\cT_4, \{(p_0, p_1)\})\}$, and $\cE_0(z^\prime) = \{(\cT_5, \{(p^{\prime}_0, p^{\prime}_1)\})\}$ nondeterministically.
\item Select the vertex $x$ in $G_0$, construct $\cE_1(y)$ and $\cE_1(z)$ as in Example~\ref{exmp-sl}, that is, nondeterministically choose $T_z =\{(q_0, q_0), (q_1,q_2)\}$, let 
$$\cE_1(z) = \{(\cT_3, \{(q^{\prime\prime}_0, q^{\prime\prime}_0)\}), (\cT_1, \{(q_0,q_0), (q_1, q_2)\})\} \mbox{ and } \cE_1(y)=\{(\cT_2, \{(q^\prime_0, q^\prime_0)\}), (\cT_{\cT_1, 0, T_z}, \{(q_0,q_2)\})\},$$ 
where $\cT_{\cT_1, 0, T_z}$ 
is the transition graph of $\cB_{\cA_1, 0, T_z}$ illustrated in Figure~\ref{fig-sl-exmp}. In addition, $\cE_1(x)=\cE_0(x)$, $\cE_1(y')=\cE_0(y')$ and $\cE_1(z')=\cE_0(z')$. Finally, we get $G_1$ from $G_0$ by removing the two edges from $x$.
\item Select the vertex $y$ in $G_1$, construct $\cE_2(y')$ and $\cE_2(z')$ as follows: Nondeterministically choose $T_{1,z'} = \{(q^\prime_0, q^\prime_0)\}$  for $\cT_2$ and $T_{2,z'}=\{(q_0,q_1),(q_1,q_2)\}$ for $\cT_{\cT_1, 0, T_z}$, let
$$\cE_2(z')=\left\{(\cT_5, \{(p^{\prime}_0, p^{\prime}_1)\}), (\cT_2, \{(q^\prime_0, q^\prime_0)\}), (\cT_{\cT_1, 0, T_z}, \{(q_0,q_1),(q_1,q_2)\}) \right\}, \text{ and}$$
$$\cE_2(y')=\left\{ (\cT_4, \{(p_0, p_1)\}), (\cT_{\cT_2, 1, T_{1,z'}}, \{(q^\prime_0, q^\prime_0)\}), (\cT_{\cT_{\cT_1, 0, T_z}, 1, T_{2,z'}}, \{(q_0,q_2)\}) \right\},$$
where $\cT_{\cT_2, 1, T_{1,z'}}$ and $\cT_{\cT_{\cT_1, 0, T_z}, 1, T_{2,z'}}$ are shown in Figure~\ref{fig-sl-exmp-nested-2}. In addition, $\cE_2(x)=\cE_1(x)$, $\cE_2(y)=\cE_1(y)$, and $\cE_2(z)=\cE_1(z)$. Finally, we get $G_2$ from $G_1$ by removing the two edges from $y$.
\end{itemize}
Since $G_2$ contains no edges, we have $\cE(x)=\cE_2(x)$, similarly for $\cE(y)$, $\cE(z)$, $\cE(y')$, and $\cE(z')$.
For the three source variables $y', z', z$, it is not hard to check that $01$ belongs to the intersection of the regular constraints in $\cE(z')$, $11$ belongs to the intersection of the regular constraints in $\cE(y')$, and $10$ belongs to the intersection of the regular constraints in $\cE(z)$. Then $y$ takes the value $\replaceall(11, 1, 01)=0101 \in \Ll(e_2)$, and $x$ takes the value $\replaceall(0101, 0, 10)=101101 \in \Ll(e_1)$. Therefore, $C$ is satisfiable. \qed
\begin{figure}[htbp]
\begin{center}
\includegraphics[scale=0.65]{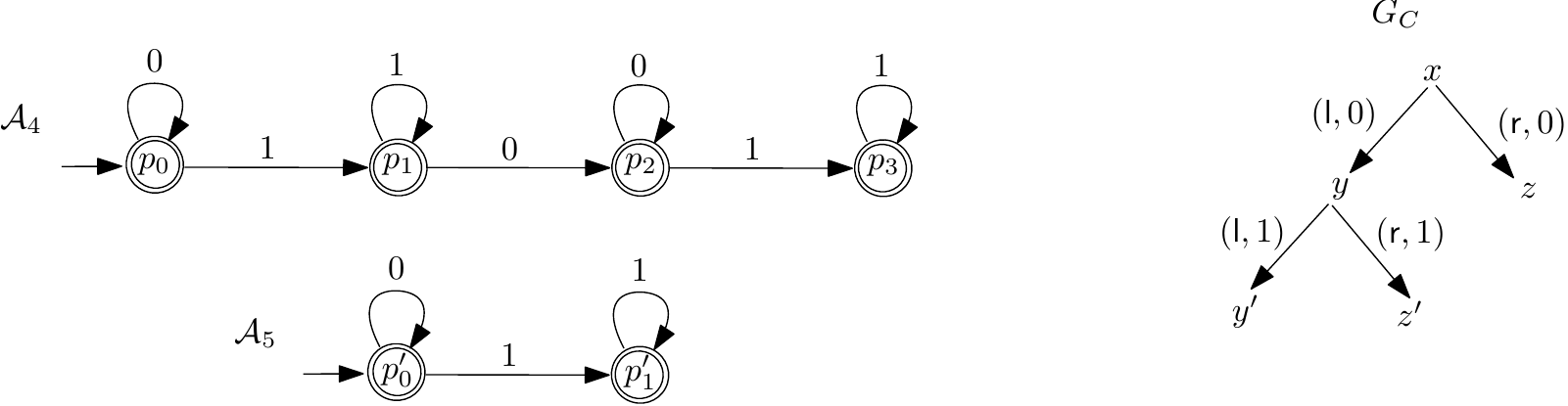}
\end{center}
\caption{An example for the single-letter case: Multiple $\replaceall$}\label{fig-sl-exmp-nested}
\end{figure}

\begin{figure}[htbp]
\begin{center}
\includegraphics[scale=0.65]{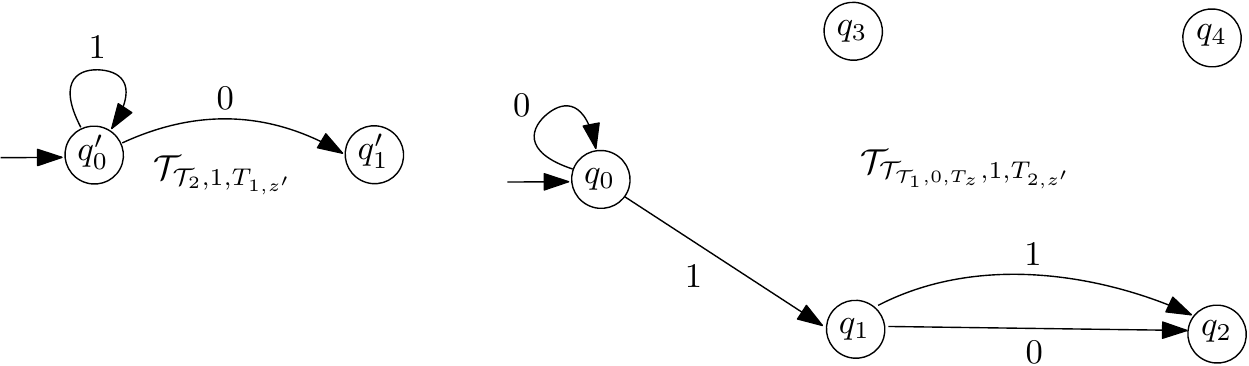}
\end{center}
\caption{$\cT_{\cT_2, 1, T_{1,z'}}$ and $\cT_{\cT_{\cT_1, 0, T_z}, 1, T_{2,z'}}$}\label{fig-sl-exmp-nested-2}
\end{figure}
\end{example}

\subsubsection{Complexity}

To show our decision procedure works in exponential space, it is sufficient to show that the cardinalities of the sets $\cE(x)$ are exponential w.r.t. the size of $C$.

\begin{proposition}\label{prop-sl-comp}
The cardinalities of $\cE(x)$ for the variables $x$ in $G_C$ are at most exponential in $\dmdidx(G_C)$, the diamond index of $G_C$.
\end{proposition}
Therefore, according to Proposition~\ref{prop-sl-comp}, if the diamond index of $G_C$ is bounded by a constant $c$, then the cardinalities of $\cE(x)$ become \emph{polynomial} in the size of $C$ and we obtain a polynomial space decision procedure. In this case, we conclude that the satisfiability problem is PSPACE-complete.

\begin{proof}[Proof of Proposition~\ref{prop-sl-comp}]
Let $K$ be the maximum of $|\cE_0(x)|$ for $x \in \vars(\varphi)$.
For each variable $x$ in $G_C$, all the regular constraints in $\cE(x)$ are either from $\cE_0(x)$, or are generated from some regular constraints from $\cE_0(x')$ for the ancestors $x'$ of $x$. Let $x'$ be an ancestor of $x$. Then for each $(\cT, \cP) \in \cE_0(x')$, according to Step I in the decision procedure, by an induction on the maximum length of the paths in from $x'$ to $x$, we can show that the number of elements in $\cE(x)$ that are generated from $(\cT, \cP)$ is at most the number of different paths from $x'$ to $x$.
From Proposition~\ref{prop-di}, we know that there are at most $(|\vars(\varphi)|\cdot |E_C|)^{O(\dmdidx(G_C))}$ different paths from $x'$ to $x$. Since there are at most $|\vars(\varphi)|$ ancestors of $x$, we deduce that $|\cE(x)| \le K \cdot |\vars(\varphi)| \cdot (|\vars(\varphi)||E_C|)^{O(\dmdidx(G_C))}$.
\end{proof}

\section{Decision procedure for $\strline[\replaceall]$: The constant-string case}\label{sec:replaceallcs}

In this section, we consider the constant-string special case, that is, for an $\strline[\replaceall]$ formula $C = \varphi \wedge \psi$, every term of the form $\replaceall(z, e, z')$ in $\varphi$ satisfies that $e=u$ for $u \in \Sigma^+$. Note that the case when $u=\epsilon$ will be dealt with in Section~\ref{sec:replaceallre}. 

Again, let us start with the simple situation that
$C \equiv x = \replaceall(y, u, z) \wedge x \in e_1 \wedge y \in e_2 \wedge z \in e_3$,
where $|u| \ge 2$. For $i=1,2,3$, let $\cA_i = (Q_i, \delta_i, q_{0, i}, F_i)$
be the NFA corresponding to $e_i$. In addition, let $k = |u|$ and $u = a_1 \cdots a_k$ with $a_i \in \Sigma$ for each $i \in [k]$.

From the semantics, $C$ is satisfiable iff $x, y, z$ can be assigned with  strings $v, w, w'$ such that: (1) $v = \replaceall(w, u, w')$, and (2) $v,w, w'$ are accepted by $\cA_1, \cA_2, \cA_3$ respectively. Let $v, w, w'$ be the strings satisfying these two constraints. Since $v = \replaceall(w, u, w')$, we know that there are strings $w_1, w_2, \cdots, w_n$ such that $w= w_1 u w_2 \cdots u w_n$ and $v = w_1 w' w_2 \cdots w' w_n$. As $v$ is accepted by $\cA_1$, there is an accepting run of $\cA_1$ on $v$, say
$$
q_{0,1} \xrightarrow[\cA_1]{w_1} q_1 \xrightarrow[\cA_1]{w'} q'_1 \xrightarrow[\cA_1]{w_2} q_2 \xrightarrow[\cA_1]{w'} q'_2 \cdots q_{n-1} \xrightarrow[\cA_1]{w'} q'_{n-1} \xrightarrow[\cA_1]{w_n} q_n.
$$
Let $T_z = \{(q_i, q'_i) \mid i \in [n]\}$. Then $w' \in \Ll(\cA_3) \cap\ \bigcap \limits_{(q,q') \in T_z} \Ll(\cA_1(q, q'))$. Therefore, $\Ll(\cA_3) \cap\ \bigcap \limits_{(q,q') \in T_z} \Ll(\cA_1(q, q')) \neq \emptyset$. Similar to the single-letter case, we construct an NFA $\cB_{\cA_1, u, T_z}$ to characterise the satisfiability of $C$.  More precisely, $C$ is satisfiable iff there is $T_z \subseteq Q_1 \times Q_1$ such that $\Ll(\cA_3) \cap\ \bigcap \limits_{(q,q') \in T_z} \Ll(\cA_1(q, q')) \neq \emptyset$ and
$\Ll(\cA_2) \cap \Ll(\cB_{\cA_1, u, T_z}) \neq \emptyset$. Intuitively, when reading the string $w$, $\cB_{\cA_1, u, T_z}$ simulates the generation of $v$ from $w$ and $w'$ (that is, the replacement of  every occurrence of $u$ in $w$ with $w'$) and verifies that $v$ is accepted by $\cA_1$, by using $T_z$.
To build $\cB_{\cA_1, u, T_z}$, we utilise the concepts of window profiles and parsing automata defined below.
Intuitively, a window profile keeps track of which positions in the preceding characters could form the beginning of a match of $u$.

\begin{definition}[window profiles w.r.t. $u$]\zhilin{I changed the notation from $k$-window profile to window profile, and $\wprof_{u,k}$ to $\wprof_u$. Please check whether there are still occurrences of the old notation below.}
Let $v$ be a nonempty string with $k=|v|$, and $i \in [k]$. Then the \emph{window profile of the position $i$ in $v$ w.r.t. $u$} is $\overrightarrow{W}  \in \{\bot,\top\}^{k-1}$ defined as follows:
\begin{itemize}
\item If $i \ge k-1$, then for each $j \in [k-1]$, $\overrightarrow{W}[j] = \top$ iff $v[i-j+1] \cdots v[i]=u[1] \cdots u[j]$.
\item If $i < k-1$, then for each $j \in [i]$, $\overrightarrow{W}[j] = \top$ iff $v[i-j+1] \cdots v[i]=u[1] \cdots u[j]$, and for each $j: i < j \le k-1$, $\overrightarrow{W}[j] = \bot$.
\end{itemize}
Let $\wprof_{u}$ denote the set of window profiles of the positions in nonempty strings w.r.t. $u$.
\end{definition}


\begin{proposition}
$|\wprof_{u}| \le |u|$.
\end{proposition}
\begin{proof}
Let $k=|u|$. For each profile $\overrightarrow{W}$, let $v$ be a nonempty string and $i$ be a position of $v$ such that for each $j \in [k-1]$, $\overrightarrow{W}[j] = \top$ iff $v[i-j+1] \dots v[i] = u[1] \dots u[j]$. Define ${\sf idx}_{\overrightarrow{W}}$ as follows: If there is $j \in [k-1]$ such that $\overrightarrow{W}[j]=\top$, then ${\sf idx}_{\overrightarrow{W}}$ is the maximum of such indices $j \in [k-1]$, otherwise, ${\sf idx}_{\overrightarrow{W}} =0$. The following fact holds for $\overrightarrow{W}$ and ${\sf idx}_{\overrightarrow{W}}$: 
\begin{itemize}
	\item for each $j': {\sf idx}_{\overrightarrow{W}} < j' \le k-1$, $\overrightarrow{W}[j']=\bot$,
	\item in addition, since $v[i-{\sf idx}_{\overrightarrow{W}}+1] \cdots v[i] = u[1] \cdots u[{\sf idx}_{\overrightarrow{W}}]$, the values of $\overrightarrow{W}[1],\cdots, \overrightarrow{W}[{\sf idx}_{\overrightarrow{W}}]$ are completely determined by $u[1] \cdots u[{\sf idx}_{\overrightarrow{W}}]$.
\end{itemize}
Let $\eta: \wprof_{u} \rightarrow \{0\} \cup [k-1]$ be a function such that for each $\overrightarrow{W} \in \wprof_{u}$, $\eta(\overrightarrow{W})={\sf idx}_{\overrightarrow{W}}$. Then $\eta$ is an injective function, since for every $\overrightarrow{W}, \overrightarrow{W'} \in \wprof_{u}$, ${\sf idx}_{\overrightarrow{W}}  = {\sf idx}_{\overrightarrow{W'}}$ iff $\overrightarrow{W} = \overrightarrow{W'}$. Therefore, we conclude that  $ |\wprof_{u}| \le k$.
\end{proof}

\begin{example}\label{wprof-exmp}
Let $\Sigma = \{0,1\}$, $u = 010$. Then $\wprof_{u}=\{\bot\bot, \top\bot, \bot\top\}$.
\begin{itemize}
\item Consider the string $v=1$ and the position $i=1$ in $v$. Since $v[1]=1 \neq u[1]=0$, the window profile of $i$ in $v$ w.r.t. $u$ is $\bot \bot$.
\item Consider the string $v=00$ and the position $i=2$ in $v$. Since $v[2]=u[1]$ and $v[1]v[2] \neq u[1]u[2]$, the window profile of $i$ in $v$ w.r.t. $u$ is $\top\bot$.
\item Consider the string $v=01$ and the position $i=2$ in $v$. Since $v[2] \neq u[1]$ and $v[1]v[2] = u[1]u[2]$, the window profile of $i$ in $v$ w.r.t. $u$ is $\bot\top$.
\end{itemize}
Note that $\top\top \not \in \wprof_{u}$, since for every string $v$ and the position $i$ in $v$, if $v[i-1]v[i]=u[1]u[2]=01$, then $v[i]=1 \neq 0= u[1]$.
\end{example}


We will construct a parsing automaton $\cA_u$ from $u$, which parses a string $v$ containing at least one occurrence of $u$ (i.e.\ $v \in \Sigma^\ast u \Sigma^\ast$) into $v_1 u v_2 u \dots v_l u v_{l+1}$ such that $v_j u[1] \dots u[k-1] \not \in \Sigma^\ast u \Sigma^\ast$ for each $1 \le j \le l$.
This ensures that the only occurrence of $u$ in each $v_j u$ is a suffix.
Finally, we also require $v_{l+1} \not \in \Sigma^\ast u \Sigma^\ast$.
The window profiles w.r.t. $u$ will be used to ensure that $v$ is correctly parsed, namely, the first, second, $\cdots$, occurrences of $u$ are correctly identified.

\begin{definition}[Parsing automata]
Given a string $u$ we define the \emph{parsing automaton} $\cA_u$ to be the NFA $(Q_u, \delta_u, q_{0,u}, F_u)$ where $q_{0,u}=q_0$ and the remaining components are given below.
\begin{itemize}
	\item  $Q_u =\left\{q_0 \right\} \cup \left\{ \left(\search, \overrightarrow{W} \right)\ \big\vert\ \overrightarrow{W} \in \wprof_{u} \right\} \cup  \left\{ \left(\verify, j, \overrightarrow{W} \right) \ \big\vert\ j \in [k-1], \overrightarrow{W} \in \wprof_{u} \right\}$, where $q_0$ is a distinguished state whose purpose will become clear later on,  and the tags ``$\search$" and ``$\verify$" are used to denote whether $\cA_u$ is in the ``search'' mode to search for the next occurrence of $u$, or in the ``verify'' mode to verify that the current position is a part of an occurrence of $u$.

	\item $\delta_{u}$ is defined as follows.
	\begin{itemize}
		\item The transition $\left(q_0, a, \left(\search, \overrightarrow{W}\right)\right) \in \delta_u$, where $\overrightarrow{W}[1]=\top$ iff $a = u[1]$, and for each $i: 2 \le i \le k-1$, $\overrightarrow{W}[i] = \bot$.
		\item The transition $\left(q_0, u[1], \left(\verify, 1, \overrightarrow{W}\right)\right) \in \delta_u$, where $\overrightarrow{W}[1]=\top$ and for each $i: 2 \le i \le k-1$, $\overrightarrow{W}[i] = \bot$.
		\item For each state $\left(\search, \overrightarrow{W} \right)$ and $a \in \Sigma$ such that $\overrightarrow{W}[k-1] = \bot$ or $a \neq u[k]$,
		\begin{itemize}
			\item the transition $\left(\left(\search, \overrightarrow{W} \right), a, \left(\search, \overrightarrow{W'} \right)\right) \in \delta_u$, where $\overrightarrow{W'}[1] = \top$ iff $a = u[1]$, and for each $i: 2 \le i \le k-1$, $\overrightarrow{W'}[i] =\top$ iff ($\overrightarrow{W}[{i-1}] = \top$ and $a = u[i]$),
			\item if $a = u[1]$, then the transition $\left(\left(\search, \overrightarrow{W} \right), a, \left(\verify, 1, \overrightarrow{W'} \right)\right) \in \delta_u$, where $\overrightarrow{W'}[1]=\top$,  and for each $i: 2 \le i \le k-1$, $\overrightarrow{W'}[i] =\top$ iff ($\overrightarrow{W}[{i-1}] = \top$ and $a = u[i]$).
		\end{itemize}
		\item For each state $\left(\verify, i-1, \overrightarrow{W} \right)$ and $a \in \Sigma$ such that
		\begin{itemize}
			\item $2 \le i \le k-1$,
			\item $\overrightarrow{W}[i-1]=\top$, $a = u[i]$, and
			\item either $\overrightarrow{W}[k-1]=\bot$ or $a \neq u[k]$,
		\end{itemize}
		we have $\left(\left(\verify, i-1, \overrightarrow{W} \right), a, \left(\verify, i, \overrightarrow{W'} \right)\right) \in \delta_u$, where for each $j: 2 \le j \le k-1$, $\overrightarrow{W'}[j] = \top$ iff $\overrightarrow{W}[j-1]=\top$ and $a = u[j]$.
		\item For each state $\left(\verify, k-1, \overrightarrow{W} \right)$ and $a \in \Sigma$ such that $\overrightarrow{W}[k-1]=\top$ and $a  = u[k]$, we have $\left(\left(\verify, k-1, \overrightarrow{W} \right), a, q_0\right) \in \delta_u$.
		%
	\end{itemize}
Note that the constraint $\overrightarrow{W}[k-1] = \bot$ or $a \neq u[k]$ is used to guarantee that each occurrence of the state $q_0$, except the first one, witnesses the \emph{first} occurrence of $u$ from the beginning or after its previous occurrence. In other words, the constraint $\overrightarrow{W}[k-1] = \bot$ or $a \neq u[k]$ is used to guarantee that after an occurrence of $q_0$, if $q_0$ has not been reached again,  then $u$ is forbidden to occur.

	%
	\item $F_u= \left\{q_0 \right\} \cup \left\{\left(\search, \overrightarrow{W} \right)\ \big\vert\ \overrightarrow{W} \in \wprof_{u} \right\} $. \\
	Note that the states $\left(\verify, j, \overrightarrow{W} \right)$ are not final states, since, when in these states, the verification of the current occurrence of $u$ has not been complete yet.
\end{itemize}
\end{definition}

Let $Q_{\search}  = \left\{ \left(\search, \overrightarrow{W} \right) \ \big\vert\ \overrightarrow{W} \in \wprof_{u} \right\}$,  and $Q_{\verify, i} = \left\{ \left(\verify, i, \overrightarrow{W} \right) \ \big\vert\ \overrightarrow{W} \in \wprof_{u} \right\}$ for each $i \in [k-1]$. In addition, let $Q_{\verify} = \bigcup \limits_{i \in [k-1]} Q_{\verify,i}$.
Suppose $v = v_1 u v_2 u \cdots v_l u v_{l+1}$ such that $v_j u[1] \dots u[k-1] \not \in \Sigma^\ast u \Sigma^\ast$ for each $1 \le j \le l$, in addition, $v_{l+1} \not \in \Sigma^\ast u \Sigma^\ast$. Then there exists a \emph{unique} accepting run $r$ of $\cA_u$ on $v$ such that the state sequence in $r$ is of the form
$q_0\ r_1\ q_0\ r_2\ q_0\ \cdots\ r_l\ q_0\ r_{l+1}$, where for each $j \in [l]$, $r_j \in  \Ll((Q_{\search})^+ \concat Q_{\verify, 1} \concat \cdots  \concat Q_{\verify, k-1})$, and $r_{l+1} \in \Ll((Q_{\search})^*)$. 


\begin{example}
Consider $u=010$ in Example~\ref{wprof-exmp}. The parsing automaton  $\cA_u$ is illustrated in Figure~\ref{fig-pa-exmp}. Note that there are no $0$-transitions out of $(\search, \bot\top)$, since this would imply an occurrence of $u = 010$, which should be verified by the states from $Q_{\verify}$, more precisely, by the state sequence $q_0 (\verify, 1, \top\bot) (\verify, 2, \bot\top)q_0$.
\begin{figure}[htbp]
\begin{center}
\includegraphics[scale=0.55]{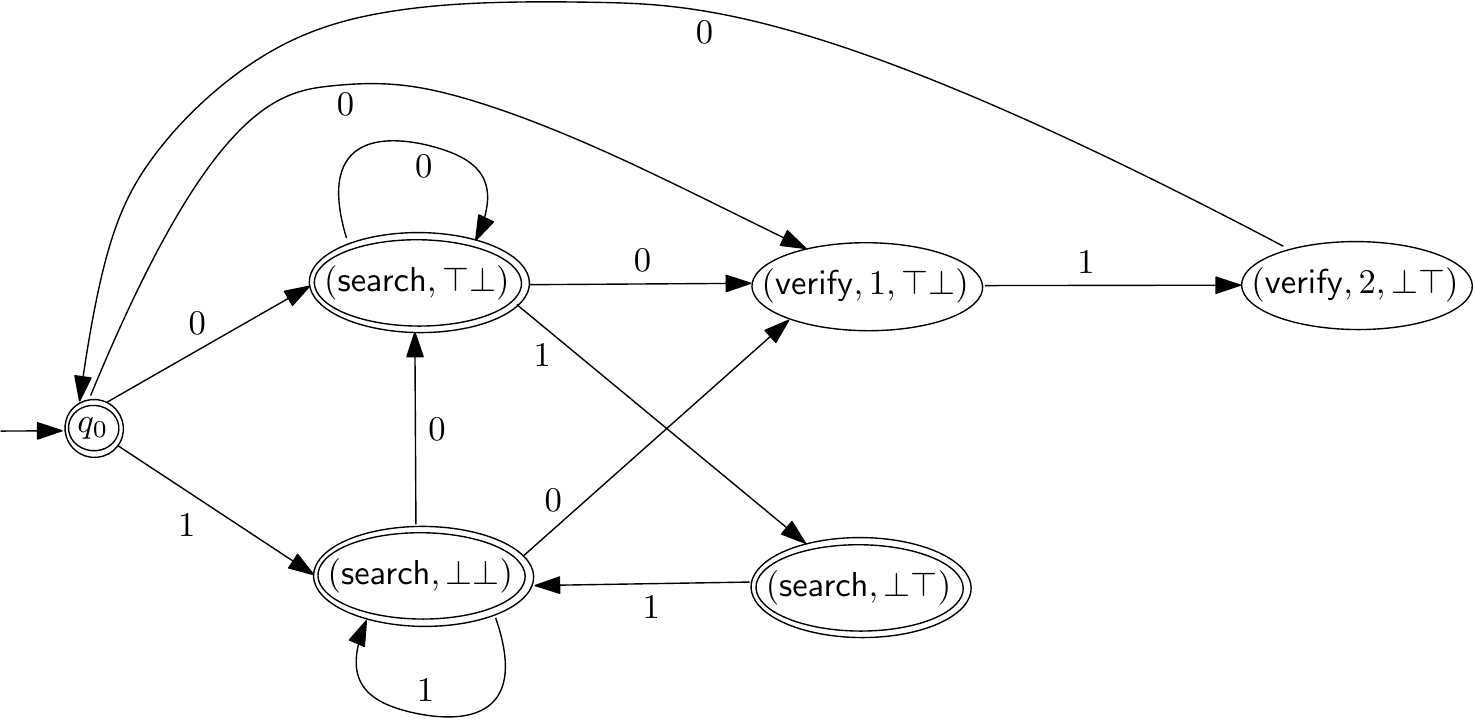}
\end{center}
\caption{The parsing automaton $\cA_u$ for $u = 010$}\label{fig-pa-exmp}
\end{figure}
\end{example}

We are ready to present the construction of $\cB_{\cA_1, u,  T_{z}}$. The NFA $\cB_{\cA_1, u, T_{z}}$ is constructed by the following three-step procedure.
\begin{enumerate}
\item Construct the product automaton $\cA_1 \times \cA_u$. Note that the initial state of $\cA_1 \times \cA_u$ is $(q_{0},q_0)$ and the set of final states of $\cA_1 \times \cA_u$ is $F_1 \times F_u$.

\item Remove from $\cA_1 \times \cA_u$ all the (incoming or outgoing) transitions associated with the states from $Q_1 \times Q_{\verify}$.

\item For each pair $(q,q') \in T_{z}$ and each sequence of transitions in $\cA_u$ of the form
$$
\left(p, u[1], \left(\verify, 1, \overrightarrow{W'_1} \right) \right), \left( \left(\verify, 1, \overrightarrow{W'_1} \right), u[2],
 \left(\verify, 2, \overrightarrow{W'_2}\right) \right), 
 \cdots, \left(\left(\verify, k-1, \overrightarrow{W'_{k-1}} \right), u[k], q_0\right),
$$
where  $p=q_0$ or $p = \left(\search, \overrightarrow{W} \right)$,
add the following transitions
$$
\begin{array}{c}
\left((q,p), u[1], \left(q, \left(\verify, 1, \overrightarrow{W'_1} \right) \right) \right), 
\left( \left(q, \left(\verify, 1, \overrightarrow{W'_1} \right) \right), u[2], \left(q, \left(\verify, 2, \overrightarrow{W'_2}\right)\right)\right),  
\cdots, \\
\left(\left(q, \left(\verify, k-2, \overrightarrow{W'_{k-2}} \right) \right), u[k-1], \left (q, \left (\verify, k-1, \overrightarrow{W'_{k-1}} \right) \right) \right),
\left( \left(q, \left (\verify, k-1, \overrightarrow{W'_{k-1}} \right) \right), u[k], \left(q', q_0 \right)\right).
\end{array}
$$
Note that the number of aforementioned sequences of transitions in $\cA_u$ is at most $|Q_{\search}|+1$, since  $ \overrightarrow{W'_1},\dots,  \overrightarrow{W'_{k-1}}$ are completely determined by $\overrightarrow{W} $ and $u$.
Intuitively, when $\cA_u$ identifies an occurrence of $u$, if the current state of $\cA_1$ is $q$, then after reading the occurrence of $u$, $\cB_{\cA_1, u, T_z}$ jumps from $q$ to some state $q'$ such that $(q,q') \in T_z$.
\end{enumerate}

\begin{example}\label{exmp-cs-case}
Consider $C \equiv x = \replaceall(y, u, z) \wedge x \in e_1 \wedge y \in e_2 \wedge z \in e_3$, where $u = 010$, and $e_1,e_2,e_3$ are as in Example~\ref{exmp-sl} (cf. Figure~\ref{fig-sl-exmp}). 
Let $T_z = \{(q_0,q_0),(q_1,q_2)\}$. The NFA $\cB_{\cA_1, u, T_z}$ is obtained
    from the product automaton $\cA_1 \times \cA_u$ (which we give in the
    \shortlong{full version}{appendix} for reference) by first removing all the transitions  associated with the states from $Q_1 \times Q_{\verify}$, then adding the transitions according to $T_z$ as aforementioned (see Figure~\ref{fig-cs-exmp-2}, where thick edges indicate added transitions).  It is routine to check that $01010101$ is accepted by $\cB_{\cA_1, u, T_z}$ and $\cA_2$. Moreover, $10 \in \Ll(\cA_3) \cap \Ll(\cA_1(q_0,q_0)) \cap \Ll(\cA_1(q_1,q_2))$. Let $y$ be $01010101$ and $z$ be $10$. Then $x$ takes the value $\replaceall(01010101, 010, 10)=101101$, which is accepted by $\cA_1$. Therefore, $C$ is satisfiable.
%
\begin{figure}[htbp]
\begin{center}
\includegraphics[scale=0.68]{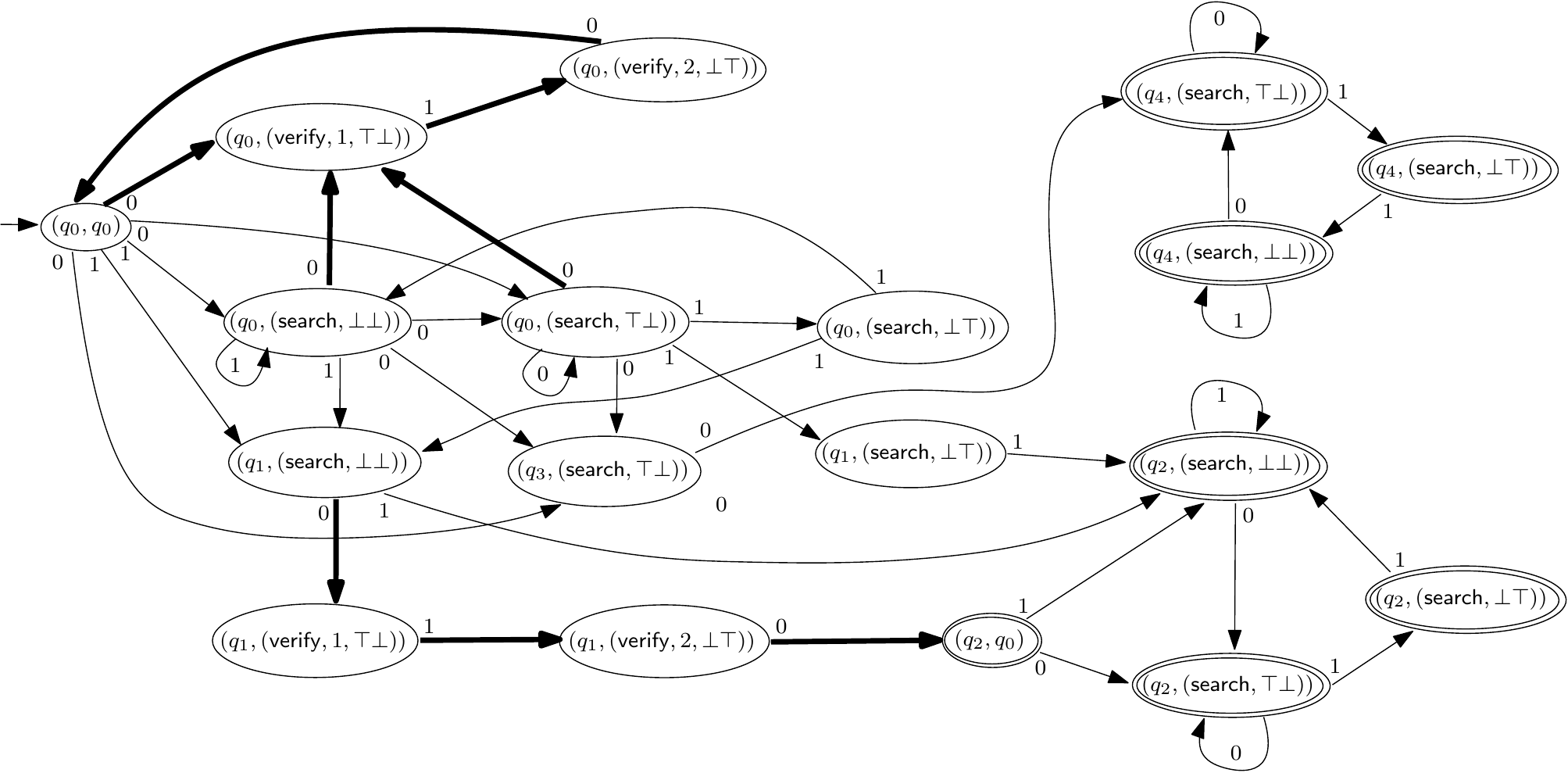}
\end{center}
\caption{The NFA $\cB_{\cA_1, u, T_z}$ for $u = 010$ and $T_z= \{(q_0,q_0),(q_1,q_2)\}$}\label{fig-cs-exmp-2}
\end{figure}
\end{example}

For the more general case that the $\strline[\replaceall]$ formula $C$ contains more than one occurrence of $\replaceall(-, -, -)$ terms, similar to the single-letter case in Section~\ref{sec:replaceallsl}, we can nondeterministically remove the edges in the dependency graph $G_C$ in a top-down manner and reduce the satisfiability of $C$ to the satisfiability of a collection of regular constraints for source variables.

\paragraph*{Complexity}
When constructing $G_{i+1}$ from $G_i$, suppose the two edges from $x$ to $y$
and $z$ respectively are currently removed, let the labels of the two edges be
$({\sf l}, u)$ and $({\sf r}, u)$ respectively. Then each element $(\cT, \cP)$
of $\cE_i(x)$ may be transformed into an element $(\cT', \cP')$ of
$\cE_{i+1}(y)$ such that $|\cT'| = O(|u||\cT|)$, meanwhile, it may also be
transformed into an element $(\cT'', \cP'')$ of $\cE_{i+1}(z)$ such that $\cT''$
has the same state space as $\cT$. In each step of the decision procedure, the
state space of the regular constraints may be multiplied by a factor $|u|$. The
state space of these regular constraints is at most exponential in the end, so
that we can still solve the nonemptiness problem of the intersection of all
these regular constraints in exponential space. In addition, if the
$\rpleft$-length of $G_C$ is bounded by a constant $c$, then for each source
variable, we get polynomially many regular constraints, where each of them has a
state space of polynomial size. Therefore, we can get a polynomial space
algorithm. See \shortlong{the full
version}{Appendix~\ref{sec:cs-complexity-full}} for a detailed analysis.




\section{Decision procedure for $\strline[\replaceall]$: The regular-expression case} \label{sec:replaceallre}

We consider the case that the second parameter of the $\replaceall$ function is a regular expression.  The decision procedure presented below is a generalisation of those in Section~\ref{sec:replaceallsl} and Section~\ref{sec:replaceallcs}.

As in the previous sections, we will again start with the simple situation that $C \equiv x = \replaceall(y, e_0, z) \wedge x \in e_1 \wedge y \in e_2 \wedge z \in e_3$. For $0\leq i\leq 3$, let $\cA_i = (Q_i, \delta_i, q_{0,i}, F_i)$ be the NFA corresponding to $e_i$.

Let us first consider the special case $\Ll(e_0)= \{\varepsilon\}$. Then according to the semantics, for each string $u = a_1 \cdots a_n$, $\replaceall(u, e_0, v) = v a_1 v \cdots v a_n v$. We  can solve the satisfiability of $C$ as follows:
\begin{enumerate}
\item Guess a set $T_z \subseteq Q_1 \times Q_1$.
\item Construct $\cB_{\cA_1, \varepsilon, T_z}$ from $\cA_1$ and $T_z$ as follows: For each $(q,q') \in T_z$, add to $\cA_1$ a transition $(q, \varepsilon, q')$. Then transform the resulting NFA into one without $\varepsilon$-transitions (which can be done in polynomial time).
\item  Decide the nonemptiness of $\Ll(\cA_2) \cap \Ll(\cB_{\cA_1, \varepsilon, T_z})$ and $\Ll(\cA_3) \cap \bigcap \limits_{(q,q') \in T_z} \Ll(\cA_1(q,q'))$.
\end{enumerate}

Next, let us assume that $\Ll(e_0) \neq \{\varepsilon\}$. For simplicity of presentation,
we assume $\varepsilon \not \in \Ll(e_0)$. The case that $\varepsilon \in \Ll(e_0)$ can be dealt with in a slightly more technical albeit similar way.

Since $\varepsilon \not \in \Ll(e_0)$, we have $q_{0,0} \not \in F_0$. In addition, without loss of generality, we assume that there are no incoming transitions for $q_{0,0}$ in $\cA_0$.

To check the satisfiability of $C$, similar to the constant-string case, we construct a parsing automaton $\cA_{e_0}$ that parses a string $v \in \Sigma^\ast e_0 \Sigma^\ast$ into $v_1 u_1 v_2 u_2 \dots v_l u_l v_{l+1}$ such that
\begin{itemize}
	\item for each $j \in [l]$, $u_j$ is the leftmost and longest matching of $e_0$ in $(v_1 u_1 \dots v_{j-1} u_{j-1})^{-1} v$,
	%
	\item $v_{l+1} \not \in \Sigma^\ast e_0 \Sigma^\ast$.
\end{itemize}

We will first give an intuitive description of the behaviour of the automaton $\cA_{e_0}$.
We start with an automaton that can have an infinite number of states and describe the automaton as starting new ``threads'', i.e., run multiple copies of $\cA_0$ on the input word (similar to alternating automata).
We also show how this automaton can be implemented using only a finite number of states.
Intuitively, in order to search for the leftmost and longest matching of $e_0$, $\cA_{e_0}$ behaves as follows.
\begin{itemize}
\item $\cA_{e_0}$ has two modes, ``$\searchleft$'' and ``$\searchlong$'', which intuitively means searching  for the first and last position of the leftmost and longest matching of $e_0$ respectively.
	\item When in the ``$\searchleft$'' mode, $\cA_{e_0}$ starts a new thread of $\cA_0$ in each position and records \emph{the set of states} of the threads into a vector.
    In addition, it nondeterministically makes a ``leftmost'' guessing, that is, guesses that the current position is the first position of the leftmost and longest matching.
    If it makes such a guessing, it enters the ``$\searchlong$'' mode, runs the thread started in the current position and searches for the last position of the leftmost and longest matching.
    Moreover, it stores in a set $S$ the union of the sets of states of all the threads that were started before the current position and continues running these threads to make sure that, in these threads, the final states will not be reached (thus, the current position is indeed the first position of the leftmost and longest matching).
	\item When in the ``$\searchlong$'' mode, $\cA_{e_0}$ runs a thread of $\cA_{0}$ to search for the last position of the leftmost and longest matching.
    If the set of states of the thread contains a final state, then $\cA_{e_0}$ nondeterministically guesses that the current position is the last position of the leftmost and longest matching.
    If it makes such a guessing, then it resets the set of states of the thread and starts a new round of searching for the leftmost and longest matching.
    In addition, it stores the original set of states of the thread into a set $S$ and continues running the thread to make sure that in this thread, the final states  will not be reached (thus, the current position is indeed the last position of the leftmost and longest matching).
	%
	%
%
%
	%
	\item Since the length of the vectors of the sets of states of the threads may become unbounded, in order to obtain a finite state automaton, the following trick is applied.
    Suppose that the vector is $S_1 S_2 \cdots S_n$.
    For each pair of indices $i, j: i < j$ and each $q \in S_i \cap S_j$, remove $q$ from $S_j$.
    The application of this trick is justified by the following arguments: Since $q$ occurs in both $S_i$ and $S_j$ and the thread $i$ was started before the thread $j$, even if from $q$  a final state can be reached in the future, the position where the thread $j$ was started \emph{cannot} be the first position of the leftmost and longest matching, since the state $q$ is also a state of the thread $i$ and the position where the thread $i$ was started is before the position where the thread $i$ was started.
\end{itemize}

Before presenting the construction of $\cA_{e_0}$ in detail, let us introduce some additional notation.

For $S \subseteq Q_0$ and $a \in \Sigma$, let $\delta_0(S,a)$ denote $\{q' \in Q_0 \mid \exists q \in S.\ (q,a,q') \in \delta_0 \}$. For $a \in \Sigma$ and a vector $\rho = S_1 \cdots S_n$ such that $S_i \subseteq Q_0$ for each $i \in [n]$, let $\delta_0(\rho,a)=\delta_0(S_1,a) \cdots \delta_0(S_n, a)$.

For a vector $S_1 \cdots S_n$ such that $S_i \subseteq Q_0$ for each $i \in [n]$, we define $\red(S_1 \cdots S_n)$ inductively:
\begin{itemize}
\item If $n = 1$, then $\red(S_1)=S_1$ if $S_1 \neq \emptyset$, and $\red(S_1)=\varepsilon$ otherwise.
\item If $n > 1$, then
\[
\red(S_1 \cdots S_n)=
\begin{cases}
\red(S_1 \cdots S_{n-1}) & \mbox{ if } S_n \subseteq \bigcup \limits_{i \in [n-1]} S_i,\\
\red(S_1 \cdots S_{n-1}) (S_n \setminus \bigcup \limits_{i \in [n-1]} S_i) &  \mbox{o/w}
\end{cases}
\]
%
\end{itemize}
For instance,
$\red(\emptyset\{q\})=\{q\}$ and
$$\red(\{q_1, q_2\} \{q_1, q_3\} \{q_2, q_4\})=\red(\{q_1,q_2\}\{q_1,q_3\})\{q_4\}= \red(\{q_1,q_2\}) \{q_3\} \{q_4\} = \{q_1,q_2\} \{q_3\}\{q_4\}.$$

We give the formal description of $\cA_{e_0}=(Q_{e_0}, \delta_{e_0}, q_{0,e_0}, F_{e_0})$ below.
The automaton will contain states of the form $(\rho, m, S)$ where $\rho$ is the vector $S_1\cdots S_n$ recording the set of states of the threads of $\cA_0$. The second component $m$ is either $\searchleft$ or $\searchlong$ indicating the mode.
Finally $S$ is the set of states representing all threads for which final states must not be reached.
\begin{itemize}
	\item $Q_{e_0}$ comprises
	\begin{itemize}
		\item the tuples $(\{q_{0,0}\}, \searchleft, S)$ such that $S \subseteq Q_0$,
		\item the tuples $(\rho \{q_{0,0}\}, \searchleft, S)$ such that  $\rho = S_1 \cdots S_n$ with $n \ge 1$ satisfying that for each $i \in [n]$, $S_i \subseteq Q_0 \setminus \{q_{0,0}\}$, and for each pair of indices $i, j: i < j$, $S_i \cap S_j = \emptyset$, moreover, $S \subseteq Q_0 \setminus F_0$,
		%
		%
		\item the tuples $(S_1, \searchlong, S)$ such that $S_1 \subseteq Q_0$, $S \subseteq Q_0 \setminus F_0$ and $S_1 \not \subseteq S$;
	\end{itemize}
	\item $q_{0,e_0}= (\{q_{0,0}\}, \searchleft, \emptyset)$,
	\item $F_{e_0}$ comprises the states of the form $(-, \searchleft, -) \in Q_{e_0}$,
	\item $\delta_{e_0}$ is defined as follows:
	\begin{itemize}
		%
		\item (continue $\searchleft$) suppose $(\rho \{q_{0,0}\}, \searchleft, S) \in Q_{e_0}$ such that $\rho = S_1 \cdots S_n$ with $n \ge 0$  ($n = 0$ means that $\rho$ is empty), $a \in \Sigma$, $\big(\bigcup \limits_{j \in [n]} \delta_0(S_j, a) \cup \delta_0(\{q_{0,0}\},a)\big) \cap F_0 = \emptyset$, and $\delta_0(S,a) \cap F_0 = \emptyset$, then

		$$\left((\rho  \{q_{0,0}\}, \searchleft, S), a, \left(\red(\delta_0(\rho\{q_{0,0}\}, a)) \{q_{0,0}\}, \searchleft, \delta_0(S,a) \right) \right) \in \delta_{e_0},$$

		\medskip

		Intuitively, in a state $(\rho, \searchleft, S)$, if $\big(\bigcup \limits_{j \in [n]} \delta_0(S_j, a) \cup \delta_0(\{q_{0,0}\},a)\big) \cap F_0 = \emptyset$ and $\delta_0(S,a) \cap F_0 = \emptyset$, then $\cA_{e_0}$ can choose to stay in the ``$\searchleft$'' mode.
		Moreover, no states occur more than once in $\red(\delta_0(\rho \{q_{0,0}\}, a)) \{q_{0,0}\}$, since $q_{0,0}$ does not occur in $\red(\delta_0(\rho\{q_{0,0}\}, a))$, (from the assumption that there are no incoming transitions for $q_{0,0}$ in $\cA_0$),
		\item (start $\searchlong$) suppose $(\rho \{q_{0,0}\}, \searchleft, S) \in Q_{e_0}$ such that $\rho = S_1 \cdots S_n$ with $n \ge 0$, $a \in \Sigma$, $\delta_0(S,a) \cap F_0 = \emptyset$, $\big(\bigcup \limits_{j \in [n]} \delta_0(S_j, a) \big) \cap F_0 = \emptyset$,  and $\delta_0(\{q_{0,0}\}, a) \not \subseteq \delta_0(S, a) \cup \bigcup \limits_{j \in [n]} \delta_0(S_j, a)$, then
		$$\left((\rho \{q_{0,0}\}, \searchleft, S), a, \left(\delta_0(\{q_{0,0}\}, a), \searchlong, \delta_0(S, a) \cup \bigcup \limits_{j \in [n]} \delta_0(S_j, a) \right) \right) \in \delta_{e_0}.$$

		Intuitively, from a state $(\rho \{q_{0,0}\}, \searchleft, S)$ with $\rho = S_1 \cdots S_n$, when reading a letter $a$, if $\big(\bigcup \limits_{j \in [n]} \delta_0(S_j, a) \big) \cap F_0 = \emptyset$, $\delta_0(S,a) \cap F_0 = \emptyset$, and $\delta_0(\{q_{0,0}\}, a) \not \subseteq \delta_0(S, a) \cup \bigcup \limits_{j \in [n]} \delta_0(S_j, a)$, then $\cA_{e_0}$ guesses that the current position is the first position of the leftmost and longest matching, it goes to the ``$\searchlong$'' mode, in addition, it keeps in the first component of the control state only the set of states of the thread started in the current position, and puts the union of the sets of the states of all the threads that have been started before, namely, $\bigcup \limits_{j \in [n]} \delta_0(S_j, a)$, into the third component to guarantee that none of these threads will reach a final state in the future (thus the guessing that the current position is the first position of the leftmost and longest matching is correct),

		\item (continue $\searchlong$) suppose $(S_1, \searchlong, S) \in Q_{e_0}$, $\delta_0(S,a) \cap F_0 = \emptyset$, and $\delta_0(S_1,a) \not \subseteq \delta_0(S,a)$, then
		$$((S_1, \searchlong, S), a, (\delta_0(S_1,a), \searchlong, \delta_0(S,a))) \in \delta_{e_0},$$
		intuitively, $\cA_{e_0}$ guesses that the current position is not the last position of the leftmost and longest matching and continues the ``$\searchlong$'' mode,
		\item (end $\searchlong$) suppose $(S_1, \searchlong, S) \in Q_{e_0}$, $\delta_0(S_1,a) \cap F_0 \neq \emptyset$, and $\delta_0(S,a) \cap F_0 = \emptyset$, then
		$$((S_1, \searchlong, S), a, (\{q_{0,0}\}, \searchleft, \delta_0(S,a) \cup \delta_0(S_1,a))) \in \delta_{e_0},$$
		intuitively, when $\delta_0(S_1,a) \cap F_0 \neq \emptyset$ and $\delta_0(S,a) \cap F_0 = \emptyset$, $\cA_{e_0}$ guesses that the current position is the last position of the leftmost and longest matching, resets the first component to $\{q_{0,0}\}$, goes to the ``$\searchleft$'' mode, and puts $\delta_0(S_1, a)$ to the third component to guarantee that the current thread will not reach a final state in the future (thus the guessing that the current position is the last position of the leftmost and longest matching is correct).
		\item ($a$ matches $e_0$) suppose $(\rho \{q_{0,0}\}, \searchleft, S) \in Q_{e_0}$ such that $\rho = S_1 \cdots S_n$ with $n \ge 0$,  $a \in \Sigma$, $\big(\bigcup \limits_{j \in [n]} \delta_0(S_j, a) \big) \cap F_0 = \emptyset$, $\delta_0(\{q_{0,0}\}, a) \cap F_0 \neq \emptyset$, and $\delta_0(S,a) \cap F_0 = \emptyset$, then
%
		$$\left(
		(\rho \{q_{0,0}\}, \searchleft, S), a,
		\left(\{q_{0,0}\}, \searchleft, \delta_0(S,a) \cup \bigcup \limits_{j \in [n]} \delta_0(S_j, a) \cup \delta_0(\{q_{0,0}\}, a) \right)
		\right) \in \delta_{e_0},$$
%
		intuitively, from a state $(\rho \{q_{0,0}\}, \searchleft, S)$ with $\rho = S_1 \cdots S_n$, when reading a letter $a$, if $\big(\bigcup \limits_{j \in [n]} \delta_0(S_j, a) \big) \cap F_0 = \emptyset$, $\delta_0(\{q_{0,0}\}, a) \cap F_0 \neq \emptyset$, and $\delta_0(S,a) \cap F_0 = \emptyset$, then $\cA_{e_0}$  guesses that $a$ is simply the leftmost and longest matching of $e_0$ (e.g. when $e_0= a$), then it directly goes to the ``$\searchleft$'' mode (without going to the ``$\searchlong$'' mode), resets the first component of the control state to $\{q_{0,0}\}$, and puts the union of the sets of the states of all the threads that have been started, including the one started in the current position, namely, $\bigcup \limits_{j \in [n]} \delta_0(S_j, a) \cup \delta_0(\{q_{0,0}\}, a)$, into the third component to  guarantee that none of these threads will reach a final state in the future (where $\bigcup \limits_{j \in [n]} \delta_0(S_j, a)$ is used to validate the leftmost guessing and $\delta_0(\{q_{0,0}\}, a)$ is used to validate the longest guessing).
	\end{itemize}
\end{itemize}

Let $Q_{\searchleft}  = \{ (-, \searchleft, -) \in Q_{e_0} \}$,  $Q_{\searchlong} = \{ (-, \searchlong, -)  \in Q_{e_0}\}$, and $v = v_1 u_1 v_2 u_2 \cdots v_l u_l v_{l+1}$ such that $u_j$ is the leftmost and longest matching of $e_0$ in $(v_1 u_1 \cdots v_{j-1} u_{j-1})^{-1} v$ for each $j \in [l]$, in addition, $v_{l+1} \not \in \Sigma^\ast e \Sigma^\ast$. Then there exists a \emph{unique} accepting run $r$ of $\cA_{e_0}$ on $v$ such that the state sequence in $r$ is of the form
%
$$
(\{q_{0,0}\}, \searchleft, \emptyset)\ r_1\ ( \{q_{0,0}\}, \searchleft, -) \ r_2\ ( \{q_{0,0}\}, \searchleft, -)
\cdots r_l\ ( \{q_{0,0}\}, \searchleft, -)\ r_{l+1},
$$
%
where for each $j \in [l]$, $r_j \in \Ll((Q_{\searchleft})^* \concat (Q_{\searchlong})^*)$, and $r_{l+1} \in \Ll((Q_{\searchleft})^*)$. Intuitively, each occurrence of the state subsequence from $\Ll((Q_{\searchlong})^* \concat (\{q_{0,0}\}, \searchleft,-))$, except the first one, witnesses the \emph{leftmost and longest} matching of $e_0$ in $v$ from the beginning or after the previous such a matching.

Since in the first component $\rho q_{0,0}$ of each state of $\cA_{e_0}$, no states from $\cA_0$ occur more than once,  it is not hard to see that $|\cA_{e_0}|$ is $2^{O(p(|\cA_0|))}$ for some polynomial $p$.



Given $T_z \subseteq Q_1 \times Q_1$, we construct $\cB_{\cA_1, e_0,  T_{z}}$ by  the following three-step procedure.
\begin{enumerate}
\item Construct the product of $\cA_1$ and $\cA_{e_0}$.

\item Remove all transitions associated with states from $Q_1 \times Q_{\searchlong}$, in addition, remove all transitions of the form $((q, (\rho\{q_{0,0}\}, \searchleft, S)), a, (q', (\{q_{0,0}\}, \searchleft,S')))$ such that $\delta_0(q_{0,0},a) \cap F_0 \neq \emptyset$.

\item For each pair $(q,q') \in T_{z}$, do the following,
\begin{itemize}
\item for each transition
%
%
$$\left((\rho \{q_{0,0}\}, \searchleft, S), a, \left(\delta_0(\{q_{0,0}\}, a), \searchlong, \delta_0(S, a) \cup \bigcup \limits_{j \in [n]} \delta_0(S_j, a) \right) \right) \in \delta_{e_0},$$
%
%
add a transition
%
%
$\left( \left(q, (\rho \{q_{0,0}\}, \searchleft, S) \right), a, \left(q, \left(\delta_0(\{q_{0,0}\}, a), \searchlong, \delta_0(S, a) \cup \bigcup \limits_{j \in [n]} \delta_0(S_j, a) \right) \right) \right),$
%
%
%
\item for each transition
		$$((S_1, \searchlong, S), a, (\delta_0(S_1,a), \searchlong, \delta_0(S,a))) \in \delta_{e_0},$$
add a transition
$\left((q, (S_1, \searchlong, S)), a, (q, (\delta_0(S_1,a), \searchlong, \delta_0(S,a))) \right),$
\item for each transition
		$$((S_1, \searchlong, S), a, (\{q_{0,0}\}, \searchleft, \delta_0(S,a) \cup \delta_0(S_1,a))) \in \delta_{e_0},$$
add a transition
		$((q, (S_1, \searchlong, S)), a, (q', (\{q_{0,0}\}, \searchleft, \delta_0(S,a) \cup \delta_0(S_1,a)))),$
\item for each
%
%
		$\left(
		(\rho \{q_{0,0}\}, \searchleft, S), a,
		\left(\{q_{0,0}\}, \searchleft, \delta_0(S,a) \cup \bigcup \limits_{j \in [n]} \delta_0(S_j, a) \cup \delta_0(\{q_{0,0}\}, a) \right)
		\right) \in \delta_{e_0},$
%
%
add a transition
%
		$$\left(
		(q, (\rho \{q_{0,0}\}, \searchleft, S)), a,
		\left(q', \left(\{q_{0,0}\}, \searchleft, \delta_0(S,a) \cup \bigcup \limits_{j \in [n]} \delta_0(S_j, a) \cup \delta_0(\{q_{0,0}\}, a) \right)\right)
		\right).$$
%
\end{itemize}
\end{enumerate}
Since $|\cA_{e_0}|$ is $2^{O(p(|\cA_0|))}$, it follows that $|\cB_{\cA_1, e_0, T_z}|$ is $|\cA_1| \cdot 2^{O(p(|\cA_0|))}$. In addition, since $|\cA_0|=O(|e_0|)$, we deduce that $|\cB_{\cA_1, e_0, T_z}|$ is $|\cA_1| \cdot 2^{O(p(|e_0|))}$.


For the more general case that the $\strline[\replaceall]$ formula $C$ contains more than one occurrence of $\replaceall(-, -, -)$ terms, we still nondeterministically remove the edges in the dependency graph $G_C$ in a top-down manner and reduce the satisfiability of $C$ to the satisfiability of a collection of regular constraints for source variables.

\paragraph*{Complexity}
In each step of the reduction, suppose the two edges out of $x$ are currently
removed, let the two edges be from $x$ to $y$ and $z$ and labeled by $({\sf l},
e)$ and $({\sf r}, e)$ respectively, then each element of $(\cT, \cP)$ of
$\cE_i(x)$ may be transformed into an element $(\cT',\cP')$ of $\cE_{i+1}(y)$
such that $|\cT'| = |\cT| \cdot 2^{O(p(|e|))}$, meanwhile, it may also be
transformed into an element $(\cT'',\cP'')$ of $\cE_{i+1}(y)$ such that $\cT''$
has the same state space as $\cT$. Thus, after the reduction, for each source
variable $x$, $\cE(x)$ may contain exponentially many elements, and each of them
may have a state space of exponential size. To solve the nonemptiness problem of
the intersection of all these regular constraints, the exponential space is
sufficient. In addition, if the $\rpleft$-length of $G_C$ is at most one, we can
show that for each source variable $x$,  $\cE(x)$ corresponds to the
intersection of polynomially many regular constraints, where each of them has a
state space at most exponential size. To solve the nonemptiness of the
intersection of these regular constraints, a polynomial space is sufficient. See
\shortlong{the full version}{Appendix~\ref{sec:re-complexity-full}} for a 
detailed analysis.



\section{Undecidable extensions}\label{sec-ext}

In this section, we consider the language $\strline[\replaceall]$ extended with either integer constraints, character constraints, or $\indexof$ constraints, and show that each of such extensions leads to undecidability. 
We will use variables of, in additional to the type $\str$, the Integer data type $\intnum$. The type $\str$ consists of the string variables as in the previous sections. A variable of type $\intnum$, usually referred to as an \emph{integer variable}, ranges over the set $\Nat$ of natural numbers. Recall that, in previous sections, we have used $x, y, z, \ldots$ to denote the variables of $\str$ type.  Hereafter we typically use $\mathfrak{l}, \mathfrak{m}, \mathfrak{n}, \ldots$ to denote the variables of $\intnum$. The
choice of omitting negative integers is for simplicity. Our
results can be easily extended to the case where $\intnum$ includes negative integers.

We begin by defining the kinds of constraints we will use to extend $\strline[\replaceall]$.
First, we describe integer constraints, which express constraints on the length or number of occurrences of symbols in words.

\begin{definition}[Integer constraints] \label{def:intconst} 
	An atomic integer constraint over $\Sigma$ is an expression of the form
	$a_1t_1+\cdots+a_nt_n\leq d$
where $a_1, \cdots, a_n,d\in \mathbb{Z}$ are constant integers (represented in binary), and each \emph{term} $t_i$ is either 
	\begin{enumerate}
		\item an integer variable $\mathfrak{n}$;
		\item $|x|$ where $x$ is a  string variable; or 
		\item $|x|_a$ where $x$ is string variable and $a\in \Sigma$ is a constant letter.
	\end{enumerate}
Here, $|x|$ and $|x|_a$ denote the length of $x$ and the number of occurrences of $a$ in $x$, respectively. 

An \emph{integer constraint} over $\Sigma$ is a Boolean combination of atomic integer constraints over $\Sigma$.
\end{definition}

Character constraints, on the other hand, allow to compare symbols from different strings. The formal definitions are given as follows. 

\begin{definition}[Character constraints]
	An \emph{atomic character constraint} over $\Sigma$ is an equation of the form $x[t_1]=y[t_2]$ where 
	\begin{itemize}
		\item $x$ and $y$ are either a string variable or a constant string in $\Sigma^*$, and 
		\item $t_1$ and $t_2$ are either integer variables or constant positive integers.
	\end{itemize} 
Here, the interpretation of $x[t_1]$ is the $t_1$-th letter of $x$. In case that $x$ does not have the $t_1$-th letter \emph{or} $y$ does not have the $t_2$-th letter, the constraint $x[t_1] = y[t_2]$ is false by convention.  
	
A \emph{character constraint} over $\Sigma$ is a Boolean combination of atomic character constraints over $\Sigma$. 
\end{definition}

We also consider the constraints involving the $\indexof$ function.

\begin{definition}[$\indexof$ Constraints]
An atomic $\indexof$ constraint over $\Sigma$ is a formula of the form $t\ \mathfrak{o}\ \indexof(s_1, s_2)$, where 
\begin{itemize}
\item $t$ is an integer variable, or a positive integer (recall that here we assume that the first position of a string is $1$), or the value $0$ (denoting that there is no occurrence of $s_1$ in $s_2$), 
\item $\mathfrak{o} \in \{\ge, \le\}$, and
\item  $s_1,s_2$ are either string variables or constant strings. 
\end{itemize}
We consider the \emph{first-occurrence} semantics of $\indexof$.  More specifically, $t \ge \indexof(s_1, s_2)$ holds if $t$ is no less than the first position in $s_2$ where $s_1$ occurs, similarly for $t \le \indexof(s_1, s_2)$.

An $\indexof$ constraint over $\Sigma$ is a Boolean combination of atomic $\indexof$ constraints over $\Sigma$.
\end{definition}

%

%
 
 \medskip

We will show that the extension of $\strline[\replaceall]$ with integer constraints entails undecidability, by a reduction from (a variant of) the Hilbert's 10th problem, which is well-known to be undecidable \cite{Mat93}. 
For space reasons, all proofs appear in \shortlong{the full
version}{Appendix~\ref{sec:ext-undec-proofs}}.
Intuitively, we want to find a solution to $f(x_1, \cdots, x_n)=g(x_1, \cdots, x_n)$ in the natural numbers, where $f$ and $g$ are polynomials with positive coefficients.
We can use the length of string variables over a unary alphabet $\{a\}$ to represent integer variables, addition can be performed with concatenation, and multiplication of $x$ and $y$ with $\replaceall(x, a, y)$.
The integer constraint $|x| = |y|$ asserts the equality of $f$ and $g$. Note that the use of concatenation can be further dispensed since, by Proposition~\ref{prop-concat},   concatenation is expressible by $\replaceall$ at the price of a slightly extended alphabet.  



\begin{theorem}\label{thm-ext-int}
	For the extension of $\strline[\replaceall]$ with \emph{integer constraints}, the satisfiability problem is undecidable, even if only a single integer constraint of the form $|x| = |y|$ or $|x|_a = |y|_a$ is used.
\end{theorem}

%

Notice that the extension of $\strline[\replaceall]$ with only one integer constraint of the form $|x| = |y|$ entails undecidability. \zhilin{a remark for the Buchi and Senger's result about undecidability, please check} \tl{rephrased slightly}
We remark that the undecidability result here does \emph{not} follow from the undecidability result for the extension of word equations with the letter-counting modalities in \cite{buchi}, 
since the formula by \cite{buchi} is not straight-line. 

By utilising a further result on Diophantine equations, we show that for the extension of $\strline[\replaceall]$ with integer constraints, even if the $\strline[\replaceall]$ formulae are simple (in the sense that their dependency graphs are of depth at most one), the satisfiability problem is still undecidable (note that no restrictions are put on the integer constraints in this case).

%


\begin{theorem}\label{thm-ext-int-strong}
	For the extension of $\strline[\replaceall]$ with integer constraints, even if $\strline[\replaceall]$ formulae are restricted to those whose dependency graphs are of depth at most one, the satisfiability problem is still undecidable.
\end{theorem}

By essentially encoding $|x|=|y|$ with \emph{character} or \emph{$\indexof$ constraints}, we show:

\begin{proposition}\label{prop-ext-ch-index}
	For the extension of $\strline[\replaceall]$ with either the character constraints or the $\indexof$ constraints, the satisfiability problem is undecidable. 
\end{proposition}


\section{Related work}\label{sec-rel}

We now discuss some related work. We split our discussion into two categories: (1) theoretical results in terms of decidability and complexity; (2) practical (but generally incomplete) approaches used in string solvers.  We emphasise work on $\replaceall$ functions as they are our focus. 

\paragraph{Theoretical Results}
We have discussed in Section \ref{sec:intro} works on string constraints with 
the theory of strings with concatenation. This research programme builds on
the question of solving satisfiability of \emph{word equations}, i.e., a
string equation $\alpha = \beta$ containing concatenation of string constants
and variables. Makanin showed decidability \cite{Makanin}, whose upper bound
was improved to PSPACE in \cite{P04} using a word compression technique. 
A simpler algorithm was in recent years proposed in \cite{J17} using
the recompression technique. The best lower bound for this problem is still NP,
and closing this complexity gap is a long-standing open problem. Decidability
(in fact, the PSPACE upper bound) can be retained in the presence of regular
constraints (e.g. see \cite{Schulz}). This can be extended to existential theory
of concatenation with regular constraints using the technique of \cite{buchi}.
The replace-all operator cannot be expressed by the concatenation operator alone. For 
this reason, our decidability of the fragment of $\strline[\replaceall]$ cannot
be derived from the results from the theory of concatenation alone.


\OMIT{
Essentially, the constraint language  $\strline[\replaceall]$ studied in this paper is word equations where the $\replaceall$ functions, which subsume the concatenation operation, are used. However, the additional straight-line restrictions are imposed.  
}

Regarding the extension with length constraints, it is still a long-standing
open problem whether word equations with length constraints is decidable, though
it is known that letter-counting (e.g. counting the number of occurrences of 0s
and 1s separately) yields undecidability \cite{buchi}. It was shown in
\cite{LB16} that the length constraints (in fact, letter-counting) can be
added to the subclass of $\strline[\replaceall]$ where the pattern/replacement 
are constants, while preserving decidability. In contrast, if we allow 
variables on the replacement parameters of formulas in $\strline[\replaceall]$,
we can easily encode the Hilbert's 10th problem with length (integer) 
constraints. 

The $\replaceall$ function can be seen as a special, yet expressive, string transformation function, aka string transducer. From this viewpoint, the closest work is~\cite{LB16}, which we discuss extensively in the introduction. Here, we discuss two further 
recent transducer models: streaming string transducers \cite{AC10} and symbolic transducers \cite{symbolic-transducer}. 

A streaming string transducer is a finite state machine where  a finite set of string variables are used to store the intermediate results for output. The $\replaceall(x, e, y)$ term can be modelled by an extension of streaming string transducers \emph{with parameters}, that is, a streaming string transducer which reads an input string (interpreted as the value of $x$), uses $y$ as a free string variable which is presumed to be read-only, and updates a string variable $z$, which stores the computation result, by a string term which may involve $y$. Nevertheless, to the best of our knowledge, this extension of streaming string transducers has not been investigated so far. 

Symbolic transducers are an extension of Mealy machine to infinite alphabets by using a variable $cur$ to represent the symbol in the current position, and replacing the input and output letters in transitions with unary predicates $\varphi(cur)$ and terms involving $cur$ respectively. Symbolic transducers can model $\replaceall$ functions \emph{when the third parameter is a constant}. Inspired by symbolic transducers, it is perhaps an interesting future work to consider an extension of the $\replaceall$ function by allowing predicates as patterns. 
For instance, one may consider the term $\replaceall(x, cur \equiv 0 \bmod 2, y)$ which replaces every even number in $x$ with $y$. 

Finally, the $\replaceall$ function is related to Array Folds Logic introduced by Daca et al \cite{DHK16}. The authors considered an extension of the quantifier-free theory of integer arrays with counting. The main feature of the logic is the \emph{fold} terms, borrowed from the folding concept in functional programming languages. Intuitively, a fold term applies a function to every element of the array to compute an output. If strings are treated as arrays over a finite domain (the alphabet), the $\replaceall$ function can be seen as a fold term. Nevertheless, the $\replaceall$ function goes beyond the fold terms considered in \cite{DHK16}, since it outputs a string (an array), instead of an integer. Therefore, the results in \cite{DHK16} cannot be applied to our setting.

\paragraph{Practical Solvers}
A large amount of recent work develops practical string solvers
including  Kaluza~\cite{Berkeley-JavaScript}, Hampi~\cite{HAMPI}, 
Z3-str~\cite{Z3-str}, CVC4~\cite{cvc4}, Stranger~\cite{YABI14}, Norn~\cite{Abdulla14}, S3 and S3P~\cite{S3,TCJ16}, and FAT~\cite{Abdulla17}.
Among them, only Stranger, S3, and S3P support $\replaceall$.  




In the Stranger tool, 
an automata-based approach was provided for symbolic analysis of PHP programs, where two different semantics $\replaceall$ were considered, namely, the leftmost and longest matching as well as the leftmost and shortest matching. Nevertheless, they focused on the abstract-interpretation based analysis of PHP programs and provided an \emph{over-approximation} of all the possible values of the string variables at each program point. Therefore, their string constraint solving algorithm is \emph{not} an exact decision procedure. In contrast, we provided a decision procedure for the straight-line fragment with the rather general $\replaceall$ function, where the pattern parameter can be arbitrary regular expressions and the replacement parameter can be variables. In the latter case,  we consider the leftmost and longest semantics mainly for simplicity, and the decision procedure can be adapted to the leftmost and shortest semantics easily.

The S3 and S3P tools also support the $\replaceall$ function, where some
progressive searching strategies were provided to deal with the non-termination
problem caused by the recursively defined string operations (of which 
$\replaceall$ is a special case). Nevertheless, the solvers are 
incomplete as reasoning about unbounded strings defined recursively is in 
general an undecidable problem.

\section{Conclusion}

\OMIT{
Since any theory of strings containing the string-replace function (even the 
most restricted version where pattern/replacement strings are both constant 
strings) is undecidable unless we impose some kind of straight-line restriction 
on the formulas, 
}
We have initiated a systematic investigation of the decidability of 
the satisfiability problem for the straight-line fragments of string 
constraints involving the $\replaceall$ function and regular constraints.
The straight-line restriction is known to be appropriate for applications in symbolic execution 
of 
string-manipulating programs \cite{LB16}. Our main result is a decision 
procedure for a large fragment of the logic, wherein the pattern parameters are
regular expressions (which covers a large proportion of the usage of the 
$\replaceall$ function in practice). Concatenation is obtained for free since concatenation can be easily expressed in this fragment.
We have shown that the decidability
of this fragment cannot be substantially extended. This is achieved by showing
that if either (1) the pattern parameters are allowed to be variables, or (2) the length
constraints are incorporated in the fragment, then we get the undecidability.
Our work clarified important fundamental issues surrounding the $\replaceall$
functions in string constraint solving and provided a novel decision procedure
which paved a way to a string solver that is able to fully support the
$\replaceall$ function. This would be the most immediate future work.

\begin{acks}
T. Chen is supported by the
\grantsponsor{}{Australian Research Council }{} under Grant No.~\grantnum{}{DP160101652}
and the
\grantsponsor{}{Engineering and Physical Sciences Research Council}{} under Grant No.~\grantnum{}{EP/P00430X/1}.
    M. Hague is supported by the
    \grantsponsor{GS501100000266}{Engineering and Physical Sciences Research Council}
                 {http://dx.doi.org/10.13039/501100000266}
    under Grant No.~\grantnum{GS501100000266}{EP/K009907/1}.
    A.~Lin is supported by the European Research Council (ERC) under the European
    Union's Horizon 2020 research and innovation programme (grant agreement no
    759969).
    Z. Wu is supported by the
    \grantsponsor{}{National Natural Science Foundation of China}{}
    under Grant No.~\grantnum{}{61472474} and Grant No. ~\grantnum{}{61572478},
    \grantsponsor{}{the INRIA-CAS joint research project ``Verification, Interaction, and Proofs''}{},
    \mat{Add your sponsors here!}\zhilin{I suggest ordering the grants according to the order of the authors}
\end{acks}

\newpage


\shortlong{}{

\newpage


\appendix

\begin{center}
{\huge Appendix} \\
{\large ``What Is Decidable about String Constraints with the ReplaceAll Function''}
\end{center}

\bigskip

We provide below proofs and examples that were omitted from the main text due to space constraints.

\hide{
\noindent {\it Proposition~\ref{prop-num-path}}.
{\it Let $G=(V,E)$ be a DAG such that the out-degree of each vertex is at most two. Then there are $n^{O(\dmdidx(G))}$ different paths  in $G$.
}

\begin{proof}
\end{proof}
}
\def\refpropundpat{\ref{prop-und-pat-var}}

\section{Proof of Proposition~\protect\refpropundpat}
\label{sec:prop-und-pat-var-proof}

We recall Proposition~\ref{prop-und-pat-var} and then give its proof.

\medskip

\noindent \textsc{Proposition}~\ref{prop-und-pat-var}
{\em The satisfiability problem of $\strline[\replaceall]$ is undecidable, if the second parameters of the $\replaceall$ terms are allowed to be variables.
}

\begin{proof}
	We reduce from the Post Correspondence Problem (PCP). Recall that the input of the problem consists of two finite lists $\alpha_{1},\ldots ,\alpha_{N}$ and $\beta_1,\ldots ,\beta_N$ of nonempty strings over $\Sigma$. A solution to this problem is a sequence of indices $(i_{k})_{1\leq k\leq K}$ with $ K\geq 1$ and $ 1\leq i_{k}\leq N$ for all $k$, such that
	$	\alpha _{{i_{1}}}\ldots \alpha _{{i_{K}}}=\beta _{{i_{1}}}\ldots \beta _{{i_{K}}}.
	$
	The PCP problem is to decide whether such a solution exists or not.

	Without loss of generality, suppose $\Sigma \cap [N] = \emptyset$ and $\$ \not \in \Sigma \cup [N]$. Let $\Sigma' = \Sigma \cup [N] \cup \{\$\}$. We will construct an $\strline[\replaceall]$ formula $C$ over $\Sigma'$ such that the PCP instance has a solution iff $C$ is satisfiable. To this end, the formula $C$ utilises the capability that the second parameter of the $\replaceall$ terms may be variables.

	Let $x_1, \cdots, x_N, y_1, \cdots, y_N, z$ be mutually distinct string variables. Then the formula $C = \varphi \wedge \psi$, where
	$$
	\begin{array}{l c l}
	\varphi & = & \bigwedge \limits_{i \in [N]} (x_i = \replaceall(x_{i-1}, i, \alpha_i) \wedge y_i = \replaceall(y_{i-1}, i, \beta_i)) \wedge  z = \replaceall(x_N, y_N, \$), \\
	\psi & = & x_0 \in (1 + \cdots + N)^+ \wedge z \in \$.
	\end{array}
	$$

	It is not hard to see that $\varphi$ is a straight-line relational constraint, thus $C$ is an $\strline[\replaceall]$ formula. Note that in $\replaceall(x_N, y_N, \$)$, the second parameter is a variable. We show that $C$ is satisfiable iff the PCP instance has a solution: $C$ is satisfiable iff there is a string $i_1 \cdots i_K \in \Ll((1 + \cdots + N)^+)$ such that when $x_0$ is assigned with $i_1 \cdots i_K$, the value of $z$ is $\$$.
	Since $z = \replaceall(x_N, y_N, \$)$ and $x_N, y_N \in \Sigma^+$, we know that $z$ is $\$$ iff the values of $x_N$ and $y_N$ are the same. Therefore, $C$ is satisfiable iff there is a string $i_1 \cdots i_K \in \Ll((1 + \cdots + N)^+)$ such that when $x_0$ is assigned with $i_1 \cdots i_K$, the values of $x_N$ and $y_N$ are the same. Therefore, $C$ is satisfiable iff there is a sequence of indices $i_1 \cdots i_K$ such that $\alpha_{i_1} \cdots \alpha_{i_K} = \beta_{i_1} \cdots \beta_{i_K}$, that is, the PCP instance has a solution.
	%
	%
	%
	%
	%
	%
	%
	%
	%
	%
\end{proof}

\def\refsecreplaceallsl{\ref{sec:replaceallsl}}

\section{Section~\protect\refsecreplaceallsl: The Correctness of the decision procedure}
\label{sec:dp-sl-correctness}

We argue that the procedure in Section~\ref{sec:dp-sl-general} is correct.
Note that Proposition~\ref{prop-sat-sl-case} removed a single $\replaceall(-,-,-)$ to obtain only regular constraints.
Each step of our decision procedure effectively eliminates a $\replaceall(-,-,-)$.
Similar to Proposition~\ref{prop-sat-sl-case}, each step maintains the satisfiability from the preceding step.

In more detail, from each $G_i$ we can define a constraint $C_i$. This constraint is a conjunction of the following atomic constraints.
\begin{itemize}
\item For each variable $x$ such that $(x, (\rpleft, a), y)$ and $(x, (\rpright,a), z)$ are the edges in $G_i$, we assert in $C_i$ that $x = \replaceall(y, a, z)$.
\item In addition, for each variable $x$ such that $\cE_i(x)$ is not empty, moreover, \emph{either $x$ is a source variable in $G_C$ (not $G_i$) or there are (incoming or outgoing) edges connected to $x$ in $G_i$}, let $e_i(x)$ be the regular expression equivalent to the conjunction of all constraints in $\cE_i(x)$ (Note that the conjunction of multiple regular expressions still defines a regular language). We assert in $C_i$ that $x \in e_i(x)$. Note that if $x$ is not a source variable in $G_C$ and there are no edges connected to $x$ in $G_i$, then the regular constraints in $\cE_i(x)$ are not included into $C_i$.
\end{itemize}


It is immediate that $C_0$ is equivalent to $C$.
We require the following proposition, which gives us the correctness of the decision procedure by induction.
Note that the final $C_i$ when exiting the loop will be a conjunction of regular constraints on the source variables.

\begin{proposition}
    For each $i$,  let the $\rpleft$-edge and the $\rpright$-edge from $x$ to $y$ and $z$ respectively be the two edges removed from $G_i$ to construct $G_{i+1}$. Then $C_i$ is satisfiable iff there are sets $T_{j, z}$ such that $C_{i+1}$ is satisfiable.
\end{proposition}

We can see the above proposition by observing that, in each step, $C_i$ is of the form
\[
    x = \replaceall(y, a, z) \wedge x \in e_i(x) \wedge y \in e_i(y) \wedge z \in e_i(z) \wedge C'
\]
where $C'$ does not contain $x$, and $C_{i+1}$ is of the form
\[
    y \in e_{i+1}(y) \wedge z \in e_{i+1}(z) \wedge C' \ .
\]
Note that $C'$ remains unchanged since only the two edges leaving $x$ are removed from $G_i$ and $\cE_{i+1}(x') = \cE_i(x')$ for all $x'$ distinct from $x$, $y$, and $z$.
First assume $y \neq z$.
Supposing $C_i$ is satisfiable, an argument similar to that of Proposition~\ref{prop-sat-sl-case} shows that there are sets $T_{j,z}$ such that the same values of $y$ and $z$ also satisfy $e_{i+1}(y)$ and $e_{i+1}(z)$.
Since $C'$ is unchanged, all $x'$ distinct from $x$, $y$, and $z$ can also keep the same value.
Thus, $C_{i+1}$ is also satisfiable.
In the other direction, suppose that there are sets $T_{j, z}$ such that $C_{i+1}$ is satisfiable. Take a satisfying assignment to $C_{i+1}$.
From the assignment to $y$ and $z$ we obtain as in Proposition~\ref{prop-sat-sl-case} an assignment to $x$ that satisfies $\replaceall(y, a, z) \wedge x \in e_i(x)$.
Furthermore, the assignments for $y$ and $z$ also satisfy $e_i(y)$ and $e_i(z)$ since $\cE_i(y)$ and $\cE_i(z)$ are subsets of $\cE_{i+1}(y)$ and $\cE_{i+1}(z)$.
Finally, since $C'$ is unchanged, the assignments to all other variables also transfer, giving us a satisfying assignment to $C_i$ as required.
In the case where $y = z$, the arguments proceed analogously to the case $y \neq z$.

\def\prodauttitle{$\cA_1 \times \cA_u$}
\def\defutitle{$u = 010$}
\section{The product automaton \protect\prodauttitle for \protect\defutitle}

In Figure~\ref{fig-cs-exmp} we give the product automaton $\cA_1 \times \cA_u$ for $u = 010$.
This is a straightforward product construction, but may be useful for reference when understanding Figure~\ref{fig-cs-exmp-2} which shows the automaton $\cB_{\cA_1, u, T_z}$ which is derived from the product.

\begin{figure}[htbp]
\begin{center}
\includegraphics[scale=0.65]{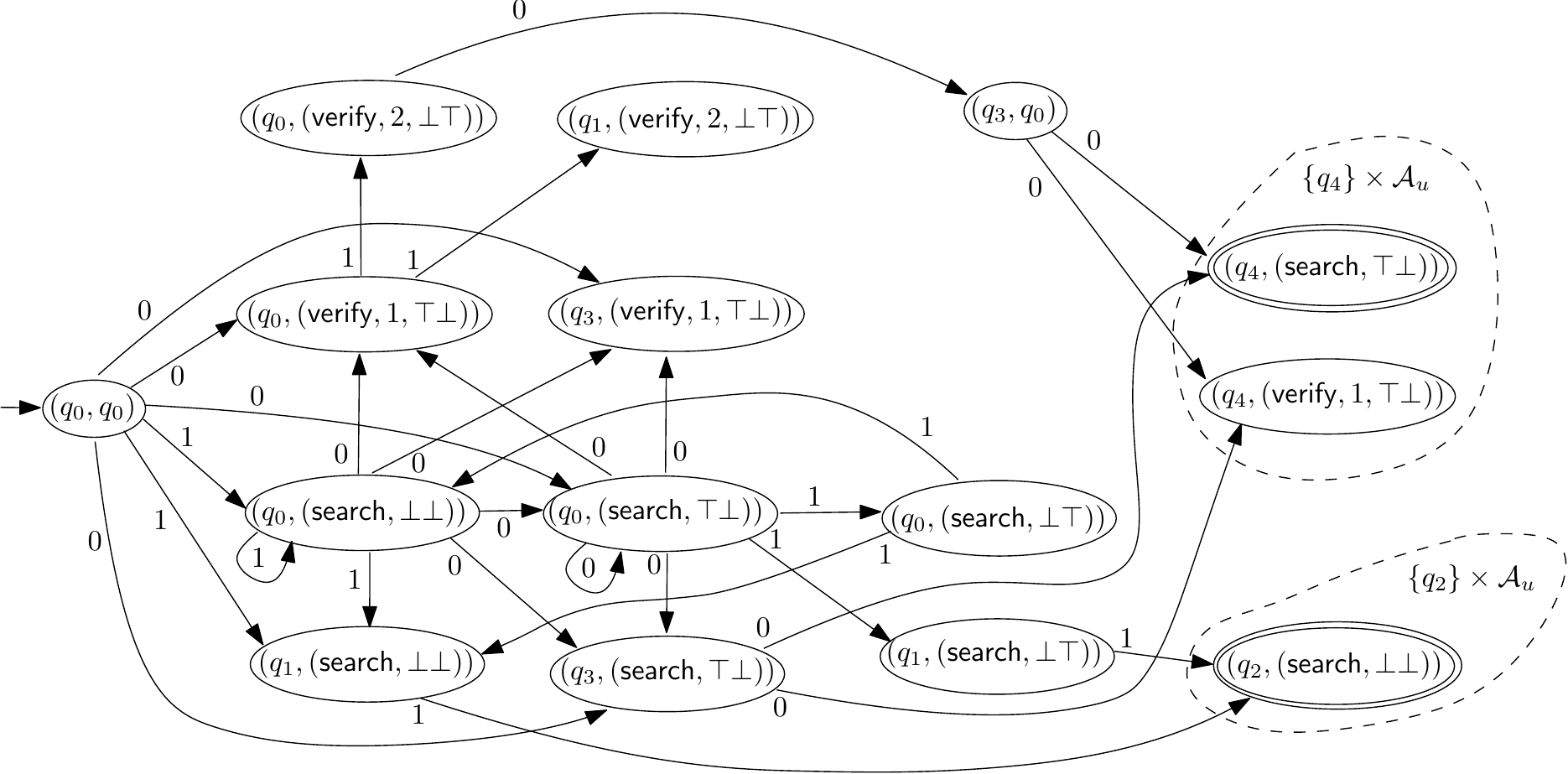}
\end{center}
\caption{The NFA $\cA_1 \times \cA_u$ for $u = 010$}\label{fig-cs-exmp}
\end{figure}

\def\refsecreplaceallcs{\ref{sec:replaceallcs}}
\section{Complexity analysis in Section~\protect\refsecreplaceallcs}
\label{sec:cs-complexity-full}

We provide a more detailed analysis of the complexity of the algorithm for the constant string case, described in Section~\ref{sec:replaceallcs}.
A summary of this argument already appears in Section~\ref{sec:replaceallcs}.

When constructing $G_{i+1}$ from $G_i$, suppose the two edges from $x$ to $y$ and $z$ respectively are currently removed, let the labels of the two edges be $({\sf l}, u)$ and $({\sf r}, u)$ respectively, then each element $(\cT, \cP)$ of $\cE_i(x)$ may be transformed into an element $(\cT', \cP')$ of $\cE_{i+1}(y)$ such that $|\cT'| = O(|u||\cT|)$, meanwhile, it may also be transformed into an element $(\cT'', \cP'')$ of $\cE_{i+1}(z)$ such that $\cT''$ has the same state space as $\cT$. Thus, for each source variable $x$, $\cE(x)$ contains at most exponentially many elements, and each of them may have a state space of at most exponential size. For instance, for a path from $x'$ to $x$ where the constant strings $u_1,\cdots, u_n$ occur in the labels of edges, an element $(\cT,\cP) \in \cE_0(x')$ may induce an element $(\cT', \cP')$ of $\cE(x)$ such that $|\cT'| \le |\cT| |u_1| \cdots |u_n|$, which is exponential in the worst case.
To solve the nonemptiness problem of the intersection of all these regular constraints, the exponential space is sufficient. Consequently, in this case, we still obtain an EXPSPACE upper-bound.

Let us now consider the special situation that the $\rpleft$-length of $G_C$ is bounded by a constant $c$.
Since $\dmdidx(G_C) \le \lftlen(G_C)$, we know that $\dmdidx(G_C)$ is also bounded by $c$. Therefore, according to Proposition~\ref{prop-di}, there are at most polynomially different paths in $G_C$, we deduce that for each source variable $x$, $\cE(x)$ contains at most polynomially many elements. In addition, since the number of $\rpleft$-edges in each path is bounded by $c$, during the execution of the decision procedure, the number of times when $(\cT, \cP)$ of $\cE_i(x)$ may be transformed into an element $(\cT', \cP')$ of $\cE_{i+1}(y)$ such that $|\cT'| = O(|u||\cT|)$ is bounded by $c$.
Therefore, for each source variable $x$ and each element $(\cT'', \cP'')$ in $\cE(x)$,  $|\cT''|$ is at most polynomial in the size of $C$. We then conclude that for each source variable $x$, $\cE(x)$ corresponds to the intersection of polynomially many regular constraints such that each of them has a state space of polynomial size. Therefore, the nonemptiness of the intersection of all the regular constraints in $\cE(x)$ can be solved in polynomial space. In this situation, we obtain a PSPACE upper-bound.

\def\refsecreplaceallre{\ref{sec:replaceallre}}
\section{Complexity analysis in Section~\protect\refsecreplaceallre}
\label{sec:re-complexity-full}

We provide a more detailed analysis of the complexity of the algorithm for the regular-expression case, described in Section~\ref{sec:replaceallre}.
A summary of this argument already appears in Section~\ref{sec:replaceallre}.

In each step of the reduction, suppose the two edges out of $x$ are currently removed, let the two edges be from $x$ to $y$ and $z$ and labeled by $({\sf l}, e)$ and $({\sf r}, e)$ respectively, then each element of $(\cT, \cP)$ of $\cE_i(x)$ may be transformed into an element $(\cT',\cP')$ of $\cE_{i+1}(y)$ such that $|\cT'| = |\cT| \cdot 2^{O(p(|e|))}$, meanwhile, it may also be transformed into an element $(\cT'',\cP'')$ of $\cE_{i+1}(y)$ such that $\cT''$ has the same state space as $\cT$. Thus, after the reduction, for each source variable $x$, $\cE(x)$ may contain exponentially many elements, and each of them may have a state space of exponential size, more precisely, if we start from a vertex $x$ without predecessors, with an element $(\cT,\cP)$ in $\cE_0(x)$, and go to a source variable $y$ through a path where $k$ edges have been traversed and removed, let $e_1,\cdots, e_k$ be the regular expressions occurring in the labels of these edges, then the resulting element in $\cE(y)$ has a state space of size $|\cT| \cdot 2^{O(p(|e_1|))} \cdot 2^{O(p(|e_2|))} \cdot \cdots \cdot 2^{O(p(|e_k|))}$ in the worst case. To solve the nonemptiness problem of the intersection of all these regular constraints, the exponential space is sufficient. Consequently, for the most general case of regular expressions, we still obtain an EXPSPACE upper-bound.

On the other hand, for the situation that the $\rpleft$-length of $G_C$ is at most one, we wan to show that the algorithm runs in polynomial space. Suppose the $\rpleft$-length of $G_C$ is at most one. Then the diamond index of $G_C$ is at most one as well. According to Proposition~\ref{prop-di}, there are only polynomially many paths in $G_C$. Nevertheless, for each source variable $x$, $\cE(x)$ may contain an element $(\cT,\cP)$ such that $|\cT|$ is exponential. Since $|\cP|$ may be exponential, $(\cT,\cP)$ may correspond to the intersection of exponentially many regular constraints. However, we can show that $|\cP|$ is at most polynomial, as a result of the fact that the $\rpleft$-length of $G_C$ is at most one. The arguments proceed as follows: Suppose two edges from $x$ to $y, z$ respectively are removed, and an element $(\cT', \cP')$ of $\cE_{i+1}(y)$ such that $|\cT'|$ is exponential and $|\cP'|$ is polynomial, is generated from an element of $(\cT, \cP)$ of $\cE_i(x)$. Then $y$ must be a source variable in $G_C$. Otherwise, there is an $\rpleft$-edge out of $y$ and the $\rpleft$-length of $G_C$ is at least two, a contradiction. Therefore, $y$ is a source variable in $G_C$, $(\cT', \cP')$  will not be used to generate the regular constraints for the other variables. In other words, $y$ is a source variable in $G_C$, and $(\cT', \cP') \in \cE(y)$ with $|\cP'|$ polynomial. We then conclude that for each source variable $x$, $|\cE(x)|$  is at most polynomial in the size of $C$ and for each element $(\cT, \cP) \in \cE(x)$, $|\cP|$ is polynomial in the size of $C$. Therefore, for each source variable $x$,  $\cE(x)$ corresponds to the intersection of polynomially many regular constraints, where each of them has a state space at most exponential size. To solve the nonemptiness of the intersection of these regular constraints, the polynomial space is sufficient. We obtain a PSPACE upper-bound for the situation that the $\rpleft$-length of $G_C$ is at most one.

\def\refsecreplaceallre{\ref{sec:replaceallre}}
\section{Examples in Section~\protect\refsecreplaceallre}

Due to space constraints, we did not provide examples of the decision procedure for the regular-expression case.
We provide some examples here.

\begin{example}\label{exmp-pa-re}
	Let $e_0 = 0^*0 1(1^* + 0^*)$. Then $\cA_{0}$ and $\cA_{e_0}$ are illustrated in Figure~\ref{fig-pa-re}, where ${\sf sleft}$ and ${\sf slong}$ are the abbreviations of $\searchleft$ and $\searchlong$ respectively. Let us use the state $(\{q_{0,1}\}\{q_{0,0}\}, {\sf sleft}, \emptyset)$ to illustrate the construction. Since $\big(\delta_0(\{q_{0,1}\}, 0) \cup \delta_0(\{q_{0,0}\}, 0)\big) \cap F_0 = \{q_{0,1}\} \cap F_0 = \emptyset$, $\delta_0(\emptyset, 0) \cap F_0 = \emptyset$, and $\red(\delta_0(\{q_{0,1}\}, 0) \delta_0(\{q_{0,0}\}, 0))=\{q_{0,1}\}$, we deduce that the transition
\[
    ((\{q_{0,1}\}\{q_{0,0}\}, {\sf sleft}, \emptyset), 0, (\{q_{0,1}\} \{q_{0,0}\}, {\sf sleft}, \emptyset)) \in \delta_{e_0} \ .
\]
On the other hand, it is impossible to go from the state $(\{q_{0,1}\}\{q_{0,0}\}, {\sf sleft}, \emptyset)$ to the ``$\searchlong$'' mode. This is due to the fact that $\delta_0(\{q_{0,0}\}, 0)=\{q_{0,1}\} \subseteq \delta_0(\{q_{0,1}\},0)=\{q_{0,1}\}$. In addition, there are no $1$-transitions out of $(\{q_{0,1}\}\{q_{0,0}\}, {\sf sleft}, \emptyset)$. This is due to the fact that $\delta_0(\{q_{0,1}\}, 1) \cap F_0 = \{q_{0,2}, q_{0,3}\} \cap F_0 \neq \emptyset$.
	\begin{figure}[htbp]
		\begin{center}
			\includegraphics[scale=0.7]{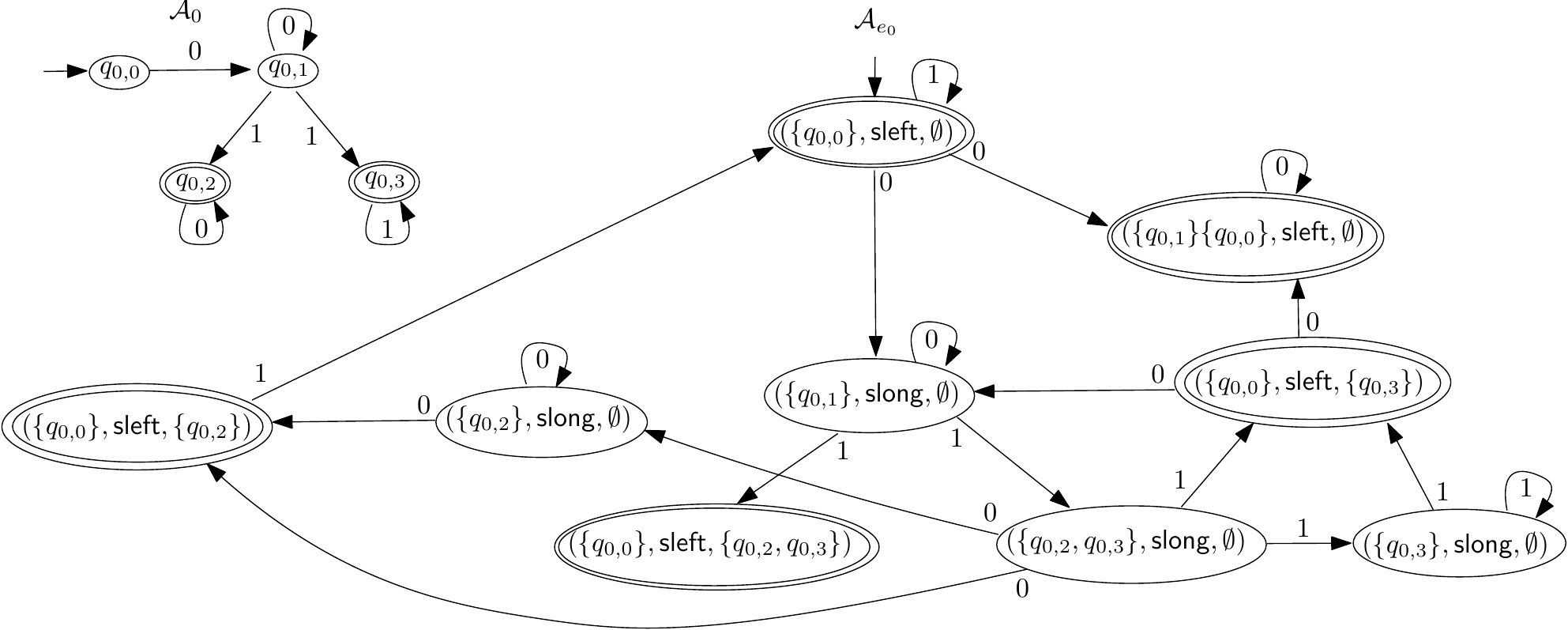}
		\end{center}
		\caption{The NFA $\cA_0$ and $\cA_{e_0}$ for $e_0 = 0^*0 1(1^* + 0^*)$}\label{fig-pa-re}
	\end{figure}
\end{example}

\begin{example}
	Let $C \equiv x = \replaceall(y, e_0, z) \wedge x \in e_1 \wedge y \in e_2 \wedge z \in e_3$, where $e_1,e_2,e_3$ are as in Example~\ref{exmp-sl} (cf. Figure~\ref{fig-sl-exmp}) and $e_0$ is as in Example~\ref{exmp-pa-re} (cf. Figure~\ref{fig-pa-re}). Suppose $T_z = \{(q_0, q_0), (q_1, q_2)\}$. Then the NFA $\cB_{\cA_1, e_0, T_z}$ is as illustrated in Figure~\ref{fig-re-exmp}, where the thick edges denote the added transitions. Let us use the state $(q_1, (\{q_{0,0}\}, \searchleft, \emptyset))$ to exemplify the construction. The transition $((q_1, (\{q_{0,0}\}, \searchleft, \emptyset)), 1, (q_2, (\{q_{0,0}\}, \searchleft, \emptyset)))$ is  in $\cA_1 \times \cA_{e_0}$. Since $\delta_0(q_{0,0}, 1) \cap F_0 = \emptyset$, this transition is not removed and is thus in $\cB_{\cA_1, e_0, T_z}$. On the other hand, since there are no $0$-transitions out of $q_1$ in $\cA_1$, there are no $0$-transitions from $(q_1, (\{q_{0,0}\}, \searchleft, \emptyset))$ to some state from $Q_{\searchleft}$ in $\cB_{\cA_1, e_0, T_z}$.
	Moreover, because $((\{q_{0,0}\}, \searchleft, \emptyset), 0, (\{q_{0,1}\}, \searchlong, \emptyset)) \in \delta_{e_0}$ and $(q_1, q_2) \in T_z$, the transition $((q_1, (\{q_{0,0}\}, \searchleft, \emptyset)), 0, (q_1, (\{q_{0,1}\}, \searchlong, \emptyset)))$ is added.
	One may also note that there are no 0-transitions from $(q_2, (\{q_{0,0}\}, \searchleft, \emptyset))$ to the state $(q_2, (\{q_{0,1}\}, \searchlong, \emptyset))$, because there are no pairs $(q2,-) \in T_z$.
	It is not hard to see that $010101 \in \Ll(\cA_2) \cap \Ll(\cB_{\cA_1, e_0, T_z})$. In addition, $10 \in \Ll(\cA_3) \cap \Ll(\cA_1(q_0,q_0)) \cap \Ll(\cA_1(q_1,q_2))$. Let $y$ be $010101$ and $z$ be $10$. Then $x$ takes the value $\replaceall(010101, e_0, 10)=10 \cdot \replaceall(101, e_0, 10)=10110$, which is accepted by $\cA_1$. Therefore, $C$ is satisfiable.
	\begin{figure}[htbp]
		\begin{center}
			\includegraphics[scale=0.68]{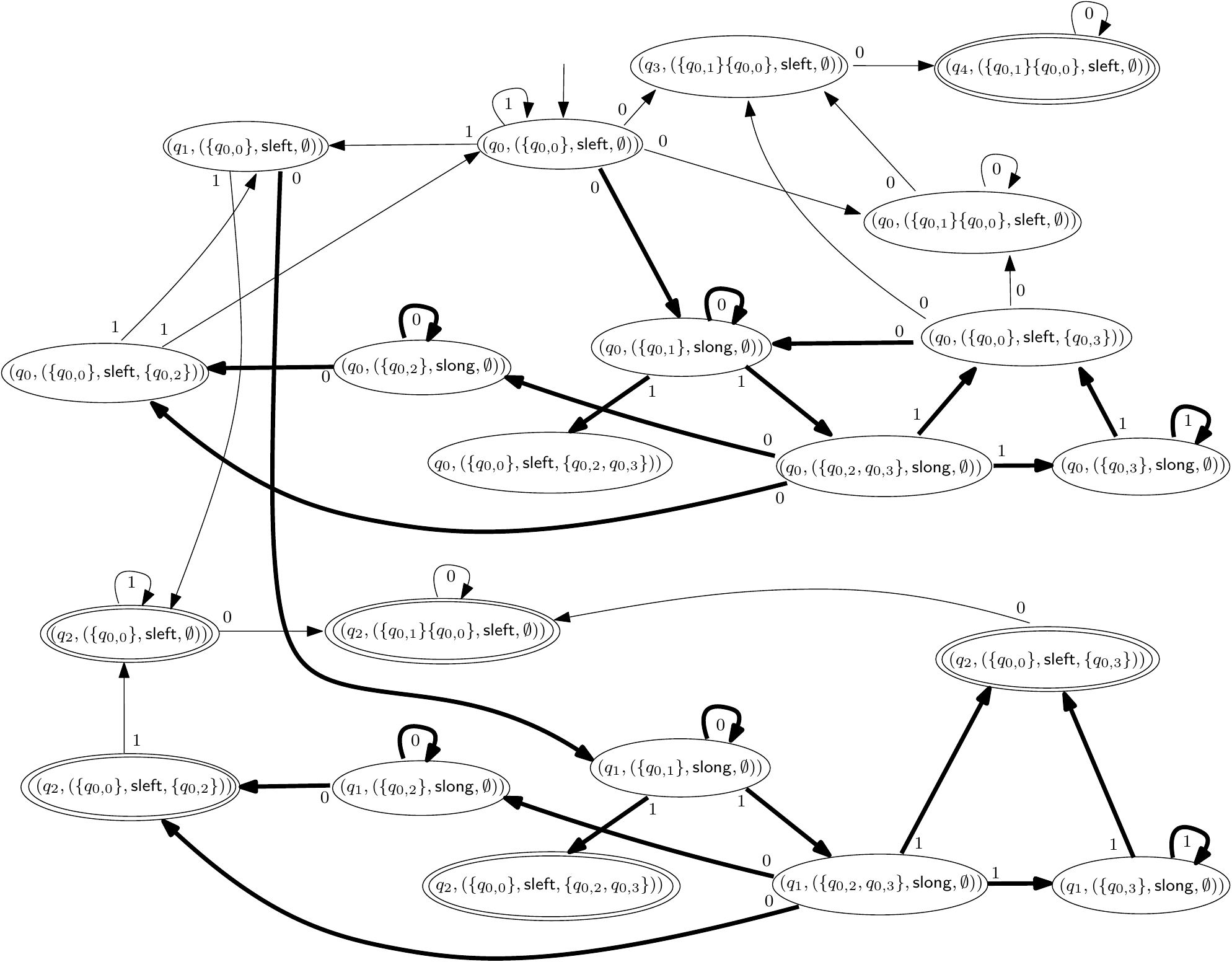}
		\end{center}
		\caption{The NFA $\cB_{\cA_1, e_0, T_z}$}\label{fig-re-exmp}
	\end{figure}
\end{example}

\def\refsecext{\ref{sec-ext}}
\section{Undecidability Proofs for Section~\protect\refsecext}
\label{sec:ext-undec-proofs}

We provide the proofs of the theorems and propositions in Section~\ref{sec-ext} which show the undecidability of various extensions of our string constraints.

\subsection{Proof of Theorem~\ref{thm-ext-int}}

We begin with the first Theorem, which is recalled below.

\medskip

\noindent \textsc{Proposition}~\ref{thm-ext-int}
{\em
    For the extension of $\strline[\replaceall]$ with \emph{integer constraints}, the satisfiability problem is undecidable, even if only a single integer constraint $|x| = |y|$ is used.
}

\begin{proof}
	The basic idea of the reduction is to simulate the two polynomials $f(x_1,\cdots, x_n)$ and $g(x_1,\cdots, x_n)$, where $x_1,\cdots,x_n$ range over the set of natural numbers, with two $\strline[\concat,\replaceall]$ formulae $C_f, C_g$ over a unary alphabet $\{a\}$, with the output string variables $y_f, y_g$ respectively, and simulate the equality $f(x_1,\cdots, x_n) = g(x_1,\cdots, x_n)$ with the integer constraint $|y_f|=|y_g|$ (which is equivalent to $y_f = y_g$, since $y_f, y_g$ represent strings over the unary alphabet $\{a\}$).

	A polynomial $f(x_1,\cdots, x_n)$ or $g(x_1,\cdots, x_n)$ where $x_1, \cdots, x_n$ range over the set of natural numbers, can be simulated by an $\strline[\concat,\replaceall]$ formula over an unary alphabet $\{a\}$ as follows: The natural numbers are represented by the strings over the alphabet $\{a\}$. A string variable is introduced for each subexpression of $f(x_1,\cdots, x_n)$. The numerical addition operator $+$ is simulated by the string operation $\concat$
	and the multiplication operator $*$ is simulated by $\replaceall$. Since it is easy to figure out how the simulation proceeds, we will only use an example to illustrate it and omit the details here. Let us consider $f(x_1,x_2) = x_1^2 + 2 x_1 x_2 + 5$. By abusing the notation, we also use $x_1,x_2$ as string variables in the simulation. We will introduce a string variable for each subexpression in $f(x_1,x_2)$, namely the variables $y_{x_1^2}, y_{x_1x_2}, y_{2x_1x_2}, y_{x_1^2+2x_1x_2}, y_{f(x_1,x_2)}$. Then $f(x_1,x_2)$ is simulated by the $\strline[\concat,\replaceall]$ formula
	\[
	\begin{array} {l c l }
	C_f & \equiv & y_{x_1^2} = \replaceall(x_1,a, x_1)\ \wedge y_{x_1x_2} = \replaceall(x_1, a, x_2)\ \wedge \\
	& & y_{2x_1x_2} = \replaceall(aa, a, y_{x_1x_2})\ \wedge y_{x_1^2+2x_1x_2} = y_{x_1^2} \concat y_{2x_1x_2}\ \wedge  \\
	& & y_{f(x_1,x_2)}=y_{x_1^2+2x_1x_2} \concat a a a a a\ \wedge x_1 \in a^*\ \wedge x_2 \in a^*.
	\end{array}
	\]
	Then according to Proposition~\ref{prop-concat}, $C_f, C_g$ can be turned into equivalent $\strline[\replaceall]$ formula $C'_f, C'_g$ by introducing fresh letters.

	Since $C'_f$ and $C'_g$ share only source variables $x_1,\cdots, x_n$, we know that $C'_f \wedge C'_g$ is still an $\strline[\replaceall]$ formula.
	From the construction of $C'_f, C'_g$, it is evident that for every pair of polynomials $f(x_1,\cdots, x_n)$ and $g(x_1,\cdots, x_n)$, $f(x_1,\cdots, x_n) = g(x_1,\cdots, x_n)$ has a solution in natural numbers iff $C'_f \wedge C'_g \wedge |y_f| = |y_g|$ is satisfiable. The proof is complete.
\end{proof}

\subsection{Undecidability of Depth-1 Dependency Graph}

We recall the undecidability of a depth-1 dependency graph before providing the proof below.

\medskip

\noindent\textsc{Theorem}\ref{thm-ext-int-strong}
{\em
	For the extension of $\strline[\replaceall]$ with integer constraints, even if $\strline[\replaceall]$ formulae are restricted to those whose dependency graphs are of depth at most one, the satisfiability problem is still undecidable.
}

\medskip

A \emph{linear polynomial} (resp.\ quadratic polynomial) is a polynomial with degree at most one (resp.\ with degree at most two) where each coefficient is an integer. 

\begin{theorem}[\cite{ID04}]\label{thm-quad-eq}
	%
	The following problem is undecidable: Determine whether a system of equations of the following form has a solution in natural numbers,
	\[
	\begin{array} {l l }
	A_i = B_i, & i =1, \cdots, k,\\
	y_iF_i=G_i \wedge y_i H_i = I_i, & i =1, \cdots, m,
	\end{array}
	\]
	where $A_i, B_i, F_i, G_i$ are linear polynomials on the variables $x_1,\cdots, x_n$ (Note that each variable $y_i$ occurs in exactly two quadratic equations).
\end{theorem}

We can get a reduction from the problem in Theorem~\ref{thm-quad-eq} to the satisfiability of the extension of $\strline[\replaceall]$ with integer constraints as follows: For each monomial $y_i x_j$ in the quadratic polynomials, we use an $\strline[\replaceall]$ formula $z_{y_i x_j} = \replaceall(y_i, a, x_j)$ to simulate $y_i x_j$, where $z_{y_i x_j}$ are freshly introduced string variables. Since each equation $y_iF_i=G_i$ or $y_i H_i = I_i$ can be seen as a linear combination of the terms $y_i x_j$ and $x_j$ for $i \in [m]$ and $j \in [n]$, we can replace each variable $x_j$ with $|x_j|$, and each term $y_ix_j$ with $|z_{y_i x_j}|$,  thus transform them into the (linear) integer constraints $F'_i = G'_i$ or $H'_i = I'_i$. Similarly, after replacing each variable $x_j$ with $|x_j|$, we transform each equation $A_i= B_i$ into an integer constraint $A'_i = B'_i$. Therefore, we get a formula
$$
\begin{array}{l c l }
\bigwedge \limits_{i \in [m], j \in [n]} z_{y_i x_j} = \replaceall(y_i, a, x_j) \wedge \bigwedge \limits_{i \in [m]} y_i \in a^*\ \wedge  \bigwedge \limits_{j \in [n]} x_j \in a^* \  \wedge\\
\hspace{2cm} \bigwedge \limits_{i \in [k]} A'_i = B'_i \wedge \bigwedge \limits_{i \in [m]} (F'_i = G'_i \wedge H'_i = I'_i),
\end{array}
$$
where the dependency graph of the $\strline[\replaceall]$ subformula is of depth at most one.

\hide{
	From this class of quadratic Diophantine equations, we can introduce string variables $x_1, \cdots, x_k$ and $y_1, \cdots, y_m$, together with relational string constraints
	\[z_{i,j}=\replaceall(x_i, a, y_j)\]
	for $1\leq i\leq k$ and $1\leq j\leq m$. Note that, for each $i$,  $t_i F_i=G_i$ can be written as
	\begin{equation} \label{eq:dio}
	t_i\cdot \left(a_0+\sum_{j=1}^s a_j s_j\right) =  b_0+\sum_{j=1}^s b_j s_j
	\end{equation}
	where $a$'s and $b$'s are all natural numbers. Moreover, \eqref{eq:dio} holds iff
	\[a_0\cdot |y_i|+ \sum_{j=1}^s a_j |z_{i,j}| =  b_0+ \sum_{j=1}^s b_j |x_j| \]
	which is an integer constraint defined in Definition~\ref{def:intconst}. This entails that
}

\subsection{Undecidability of the Character Constraints}

We provide part of the proof of Proposition~\ref{prop-ext-ch-index}, in particular, we show the undecidability of character constraints.

\begin{proposition}\label{prop-ext-char}
	For the extension of $\strline[\replaceall]$ with character constraints, the satisfiability problem is undecidable.
\end{proposition}

The arguments for Proposition~\ref{prop-ext-char} proceed as follows. Recall that in the proof of Theorem~\ref{thm-ext-int}, we get a formula $C_f \wedge C_g \wedge |y_f| = |y_g|$ such that $f(x_1,\cdots, x_n) = g(x_1,\cdots, x_n)$ has a solution in natural numbers iff $C_f \wedge C_g \wedge |y_f| = |y_g|$ is satisfiable. Let $\$ \neq a$. Suppose  $z_f = y_f \concat \$$, and $z_g = y_g \concat \$$. Then $|y_f| = |y_g|$ can be captured by $z_f[\mathfrak{n}] = \$[1] \wedge  z_g[\mathfrak{n}] = \$[1]$, where $\mathfrak{n}$ is a variable of type $\intnum$. More precisely,
we have
\begin{quote}
	\centering
	$C_f \wedge C_g \wedge |y_f|= |y_g|$ is satisfiable \\
	iff \\
	$C_f \wedge C_g \wedge z_f = y_f \concat \$ \wedge z_g = y_g \concat \$ \wedge z_f[\mathfrak{n}] = \$[1] \wedge  z_g[\mathfrak{n}] = \$[1]$ is satisfiable.
\end{quote}
Therefore, we get a reduction from Hilbert's tenth problem to the satisfiability problem for the extension of $\strline[\replaceall]$ with character constraints.

%
%
%
%

\subsection{Undecidability of the $\indexof$ Constraints}

We provide the final part of the proof of Proposition~\ref{prop-ext-ch-index}, in particular, we show the undecidability of $\indexof$ constraints.

\begin{proposition}\label{prop-indexof}
	For the extension of $\strline[\replaceall]$ with the $\indexof$ constraints, the satisfiability problem is undecidable.
\end{proposition}

Proposition~\ref{prop-ext-char} follows from the following observation and Theorem~\ref{thm-ext-int}: For any two string variables $x,y$ over a unary alphabet,
$1= \indexof(x,y)$ iff $x$ is a prefix of $y$. Therefore, $|x| = |y|$ iff $1=  \indexof(x,y) \wedge 1= \indexof(y,x)$. This implies that in the proof of Theorem~\ref{thm-ext-int}, we can replace $|y_f| = |y_g|$ with $1=\indexof(y_f, y_g) \wedge 1 = \indexof(y_g, y_f)$ and get a reduction from Hilbert's tenth problem to the satisfiability problem for the extension of $\strline[\replaceall]$ with the $\indexof$ constraints.
Note that $=$ can be simulated as a conjunction of $\leq$ and $\geq$.

}

\end{document}